\documentclass[draftclsnofoot, onecolumn, 12pt]{IEEEtran}
\usepackage{amsthm}
\usepackage{amsfonts}%
\usepackage{amsmath}%
\usepackage{fixmath}%
\usepackage{amssymb}%
\usepackage{graphicx}
\usepackage{epsfig,makeidx,color}
\usepackage[table]{xcolor}
\usepackage{comment}
\usepackage{url}
\usepackage{booktabs}
\usepackage{cite}
\usepackage{color}
\usepackage{amsthm}
\usepackage{subfloat}
\usepackage{subfig}
\usepackage{flushend}
\usepackage{bm}

\newtheorem{theorem}{Theorem}
\newtheorem{lemma}{Lemma}

\newcommand{\Eb}{\mathbb{E}}
\newcommand{\E}[1]{\mathbb{E}\left[#1\right]}

\allowdisplaybreaks

\newcommand{\var}{\text{var}}

\begin{document}
\title{Performance Analysis of Cell Free Massive MIMO systems in LoS/ NLoS Channels}
\author{\IEEEauthorblockN{Sudarshan Mukherjee, Ribhu Chopra} 
		
		\thanks{R. Chopra and S. Mukherjee  are with the Department of Electronics and Electrical Engineering, Indian Institute of Technology Guwahati, Assam, India. (email: {ribhufec, smukherjee}@iitg.ac.in).	
			
						}	
		 }
\maketitle
\begin{abstract}
In cellular communication systems, it is conventional to assume the absence of a line of sight (LoS) path between the users and their associated access points (APs). This assumption however becomes questionable in the context of recent developments in the direction of cell free (CF) massive MIMO systems. In the CF massive MIMO, the AP density is assumed to be comparable with the user density, which increases probability of existence of an LoS path between the users and their associated APs. In this paper, we analyze the performance of an uplink CF massive MIMO system, with a probabilistic LoS channel model. Here, we first derive the effective statistics of this channel model, and argue that their behaviour is fundamentally different from that of the conventional rich scattering channels. Utilizing these statistics, we next compare the rates achievable by CF massive MIMO systems, under both stream-wise and joint decoding at the central processing unit. Following this, we also discuss the centralized MMSE based data detection to obtain a complexity/ performance trade-off. Finally, using detailed Monte-Carlo simulations, we validate our analytical results, and evaluate the performance of the three data detection schemes.

\end{abstract}
\section{Introduction}
\subsection{Motivation}
{The idea of using a large number of co-located antennas, typically on a base station~(BS) to simultaneously serve a smaller number of user equipments~(UEs) over the same time frequency resources has generated much research interest over the past decade~\cite{Marzetta_TWC_2010}. This architecture consisting of a large number of antennas~(of the order of hundreds) to serve a few tens of users is dubbed as massive multiple-input-multiple-output~(MIMO)~\cite{Chokes_BSR_LMIMO,Red_book}.} The key benefits of massive MIMO include greater spectral and energy efficiencies~\cite{scaling} coupled with analytical tractability. Additionally, these advantages of massive MIMO can be achieved via employing simple linear signal processing techniques at the BS~\cite{CHEMP}. In view of these developments, massive MIMO has been widely accepted as a key enabling technology for the fifth generation~(5G) wireless systems~\cite{reality}. At the same time, due to their co-located nature, masive MIMO systems inherit the problem of non-uniform coverage of cellular users from the previous generations of cellular systems. This makes massive MIMO systems inherently unfair to the UEs located at the cell edges, or physically far away from the BSs. In view of this limitation, the research focus has recently shifted to a distributed architecture for massive MIMO systems, which eliminates the notion of cellular structure from the network. This new architecture, that has emerged as a front-runner among the various suggested solutions, is dubbed as the cell-free~(CF) massive MIMO~\cite{Ngo_TWC_2017}.

The key idea in CF massive MIMO systems is to have a large number of access points~(APs), each equipped with either a single or a small number of antennas, and spread over a large geographical area, to simultaneously serve few tens of users over the same time frequency resources. Here, the total AP antenna density is considered to be much larger than the UE density.  Additionally, all the APs are assumed to be connected to a central processing unit~(CPU) via a high rate backbone link~\cite{Ngo_TWC_2017}. Consequently, all the information available at the APs can be accurately shared with the CPU. CF massive MIMO systems have been known to inherit the advantages of conventional cellular massive MIMO systems, such as power scaling and simple linear processing, in addition to offering a near uniform coverage to all the UEs~\cite{Ngo_SPAWC_2015}.

{Most of the existing literature on the CF massive MIMO systems assumes either Rayleigh fading~\cite{Ngo_TWC_2017,Shamai_VTC_2001,Ngo_TGCN_2018,Nayebi_TWC_2017,Bjornson_TWC_2020,Bashar_ICC_2018,Dhillon_CFMM,Zhang_Access_2018,Papa_TVT_2020,making_CF} or Rician fading channels~\cite{Ozdogan_TWC_2019,Zhang_Comml_2002,Jin_sys_2020}. These assumptions, however, correspond to either of the two extremes of channel fading modelling. Note that, due to the geographically distributed nature of APs and various blockages (e.g. buildings in urban areas, trees, hillocks etc.), it is impossible to accurately determine the presence of a line of sight~(LoS) link between the APs and the UEs. Therefore, it is important that we characterize the performance of CF massive MIMO systems, under probabilistic LoS components. Additionally, the changes in channel fading distribution due to the presence/absence of an LoS link between APs and UE is known to significantly impact the overall system performance~\cite{Atzeni_TWC_2018}. }

\subsection{Prior Work}
{The idea of using only the locally available channel state information (CSI) to serve a large number of users in CF massive MIMO is based on the concept of network MIMO~\cite{Shamai_VTC_2001}. In~\cite{Ngo_TWC_2017,Ngo_TGCN_2018}, a simplistic version of CF massive MIMO with an infinite rate instantaneous backhaul link was considered, and it was shown that even with local conjugate beamforming, CF massive MIMO can out-perform both the small cell network model, and co-located massive MIMO model in terms of throughput and energy efficiency respectively. In~\cite{Ngo_TWC_2017}, it was also argued that it is impractical to forward CSI from all APs to the CPU, and was suggested that the CPU can perform data detection based on pre-processed samples from the APs. Later in~\cite{Nayebi_TWC_2017} it was shown that with an infinite capacity backhaul, the CPU can acquire CST for all channels, and use the centralized minimum mean squared error~(MMSE) combining for data detection. The issue of impracticality of infinite rate backhaul was raised in~\cite{Bashar_ICC_2018} and~\cite{Dhillon_CFMM}, where the effects of quantization in the backhaul link were derived. The trade-off between performance and backhaul requirements of conjugate and MMSE beam-forming was later discussed in detail in~\cite{making_CF}. Finally, the effects of hardware imperfections in CF massive MIMO have been discussed in~\cite{Zhang_Access_2018}. }

It is common to assume that the advantages of channel hardening and favorable propagation, enjoyed by cellular massive MIMO systems are also inherited by CF massive MIMO systems~\cite{Ngo_TWC_2017}. However, it was argued in~\cite{Hardening_CF} that due to the distributed nature of the system, the channels between different APs and users are non-identically distributed. This leads to CF massive MIMO inheriting these properties only when the number of antennas at individual APs is large enough. It has also been argued in~\cite{Papa_TVT_2020} that given the geographically distributed nature of both APs and the UEs, stochastic geometry can be used to accurately characterize the behaviour of CF massive MIMO systems.

In  recent years, the performance of conventional sparse and dense cellular networks has been analyzed under probabilistic LoS/NLoS links between the BSs and UEs~\cite{Ding_TWC_2016,Ding_TWC_2017,Cho_Comml_2018}. It has been shown that the effect of LoS/NLoS channels is more pronounced in the case of dense cellular networks, in comparison to the conventional sparse networks, where the performance is dominated by the path loss component.  It has also been shown via repeated analytical modelling and simulations that in case of dense cellular networks, the network coverage probability first increases with the AP density, and then begins to decrease, when the AP density crosses a threshold value~\cite{Ding_TWC_2017,Cho_Comml_2018}. Furthermore, depending on the accuracy and effectiveness of the modelling for LoS/NLoS channels, the predicted behaviour of the area spectral efficiency (ASE) performance of the system has also been observed to change~\cite{Ding_TWC_2016,Ding_TWC_2017}. 

In the recent recommendations of international telecommunications union (ITU), it has been argued that the existence of LoS links between APs and UE not only depends on their horizontal separation, but also on their individual antenna heights, blockage distribution, as well as on the distribution of blockage height~\cite{Bai_TWC_2014, Hourani_WCL_2014,3GPPS_1,IMT2020propagate}. Therefore, without thorough considerations for the above mentioned factors that affect the existence of LoS links, it would be difficult to truly characterize the design parameters that optimize the network performance at higher AP densities. A similar idea is also applicable for CF massive MIMO systems. However, to the best of our knowledge, no such attempt has been made so far to appropriately characterize the effect of LoS/NLoS links in CF massive MIMO systems in an consolidated manner. 

\subsection{Contributions}
In this paper, we model the channel between the APs and UEs in a CF massive MIMO system accounting for the presence of LoS links. We analyze the uplink performance of the aforementioned system under the assumption of the availability of accurate CSI at the CPU, as well as under that of MMSE channel estimation by the APs. In addition to the conventional approach of analyzing the achievable rates where each users' data is processed separately at the CPU, we also discuss the joint processing of the user data at the CPU. This approach, though prohibitively complex, provides an upper bound on the performance of such a system, and motivates us to analyze the performance of centralized MMSE based data detection. We summarize our main contributions as follows:
\begin{enumerate}
	\item We model the channel between the APs and the UEs with a probabilistic LoS component, and derive bounds on the rates achievable in the uplink under conjugate beamforming based data detection with distributed processing.  (See Sections \ref{sec:model} and \ref{sec:stream_wise}.)
	\item We derive bounds on the uplink rates achievable under joint decoding of the user data at the CPU, under both accurate and estimated CSI at the CPU. (See Section \ref{sec:joint})
	\item Motivated by the significant performance gap in the rates achievable using joint and distributed data detection, and their respective computational requirements, we discuss centralized MMSE based data detection at the CPU (See Section \ref{sec:message}).
	\item Via detailed simulations, we validate our derived results and compare the performances of the above three data detection techniques in the presence of both accurate and estimated CSI. We also discuss the effect of various system parameters such as UE/ AP density, on the overall performance of the system, and use these to prescribe the values of different controllable parameters (see Section~\ref{sec:results}). 
\end{enumerate}
The key takeaway of this work is that the performance of CF massive MIMO systems varies significantly in the presence of LoS/ NLoS channels, and the pilot and power allocation techniques devised for these systems need to be revisited. It is also interesting to note the gap between the performance of these systems under joint and stream-wise processing, indicating the need for better interference cancellation algorithms.

\begin{table}
\caption{Notations used throughout the paper.} \label{table:notation}

\begin{center}
\rowcolors{2}%{green!10!yellow}{}
{cyan!15!}{}
\renewcommand{\arraystretch}{1.0}
\begin{tabular}{| c | p{6.5cm} || c | p{6.5cm} |}
\hline 
 {\bf Notation} & {\hspace{2.5cm}}{\bf Definition} & {\bf Notation} & {\hspace{2.5cm}}{\bf Definition}
\\
\midrule
\hline
$\ell_m$ & Height of the $m$-th AP & $\ell_k^{'}$ & Height of the $k$-th UE \\ %\addlinespace
$\delta_{mk}$ & Index variable for LoS channel between the $m$-th AP and $k$-th UE & $P_{mk} \triangleq \E{\delta_{mk}}$ & Probability of the LoS channel between the $m$-th AP and $k$-th UE \\ %\addlinespace
$\bar{\mathbf{h}}_{mk}$ & LoS channel gain between the $m$-th AP and $k$-th UE & $\dot{\mathbf{h}}_{mk}$ & NLoS channel gain between the $m$-th AP and $k$-th UE \\ %\addlinespace
$x_{mk}$ & $3$-D link distance between the $m$-th AP and $k$-th UE & $d_{mk}$ & Horizontal (2D) distance between the $m$-th AP and $k$-th UE \\ %\addlinespace
$N$ & Number of antennas at the APs & $d$ & Spacing between antennas at the APs \\ %\addlinespace
$M$ & Number of APs in the system & $K$ & Number of UEs in the system \\ %\addlinespace
$d_0$ & Reference distance between any transceiver pair & $\eta$ & Pathloss exponent \\ %\addlinespace
$\alpha$ & Fraction of built-up area in the network & $\gamma$ & Average altitude of building/blockage \\ %\addlinespace
$\mu$ & Average number of blockages per unit area & $\lambda_c$ & Carrier wavelength \\ %\addlinespace
$\mathcal{E}_p$ & Average transmit pilot power & $\psi_k[n]$ & Orthogonal pilot for the $k$-th UE at $n$-th instance\\ %\addlinespace
$\mathcal{E}_{s,k}$ & Average transmit symbol power from the $k$-th UE & $\E{.}$ & Expectation operation \\ %\addlinespace
$\otimes$ & Kronecker product & $\odot$ & Hadamard product \\ %\addlinespace
$\mathcal{S}$ & Set of transmission symbols & $\left| \mathcal{S} \right|$ & Cardinality of set $\mathcal{S}$ \\ %\addlinespace
$\mathbf{O}_N$ & Null matrix of order $N$ & $\mathbf{I}_N$ & Identity matrix of order $N$ \\ %\addlinespace
$\mathbf{0}_N$ & Null vector of length $N$ & $\iota $ & $\sqrt{-1}$ \\ %\addlinespace
\hline 
\end{tabular}
\end{center}
\end{table}%

\section{System Model}\label{sec:model}
\begin{figure}[t]
	\centering
	\includegraphics[width=0.7\textwidth]{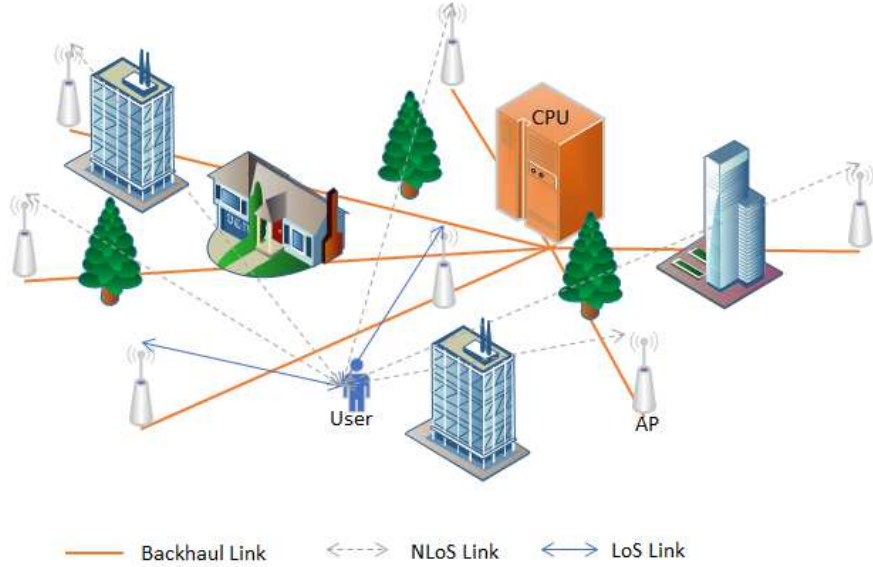}
	\caption{The System Model}
	\label{fig:model}
\end{figure}
We consider a CF massive MIMO system with $M$ APs, each equipped with a uniform linear array of $N$ antennas, serving a total of $K \ll MN$ user equipments~(UEs). The $m$th AP is assumed to be at a height $\ell_m$, with its antennas separated by a distance $d$. Similarly, the $k$th UE is assumed to be situated at a height $\ell_k^{\prime}$~(See Fig.~\ref{fig:model}.)	
We note that the line-of-sight (LoS) path between a UE and an AP may be obstructed due to the presence of blockages (e.g. buildings in urban areas, trees/ hillocks in rural areas, etc.). However, due to the random nature of the locations of UEs and APs, the existence of LoS paths between these cannot be guaranteed. Therefore, in this work, we characterize the channel between a UE and an AP as a combination of both LoS and NLoS (non-line-of-sight) channels. As a result, the channel between the $m$th AP and the $k$th UE, $\mathbf{h}_{mk}\in\mathbb{C}^{N\times1}$,  is given by
\begin{equation}
\mathbf{h}_{mk}=\delta_{mk}\bar{\mathbf{h}}_{mk}+\sqrt{\beta_{mk}}\mathbf{\dot{h}}_{mk}.
\label{eq:channelcombo}
\end{equation}
\noindent Here $\mathbf{\dot{h}}_{mk}$ represents the fast fading component of the NLoS channel, that is assumed to consist of independent and identically distributed~(i.i.d.) zero mean circular symmetric complex Gaussian (ZMCSCG) entries, i.e., $\mathbf{\dot{h}}_{mk} \sim \mathcal{C}\mathcal{N}(0, \textbf I_N)$. The coefficient $\beta_{mk}\triangleq\min\left\{1,\left(\frac{d_{mk}}{d_0}\right)^{-\eta}\right\}$ denotes the standard pathloss component of the NLoS channel, with $d_0$ being the modelled reference distance, $d_{mk}$ denoting the horizontal distance between the $m$th AP and the $k$th UE, and $\eta(>2)$ being the path loss exponent. Here, $\delta_{mk}$ is a binary valued random variable, that indicates the presence or the absence of an LoS component between the $m$th AP and $k$th UE, such that  $\Pr\{\delta_{mk}=1\}=P_{mk}$.
In \eqref{eq:channelcombo}, $\bar{\mathbf{h}}_{mk}$ denotes the LoS channel, that depends on the three dimensional link distance between the $k$th UE and the $m$th AP, $x_{mk} \triangleq \sqrt{d_{mk}^2 + (\ell_m - \ell_k^{\prime})^2}$, as well as on the angle of arrival~(AoA) at the AP, $\theta_{mk}$. Thus, $\bar{\mathbf{h}}_{mk}$ is expressed as, 
\begin{equation}
\bar{\mathbf{h}}_{mk}=\mathbf{a}(\theta_{mk})\sqrt{{G_mG_k}}\left(\frac{\ell_{k}^{'}\ell_m}{4\pi x_{mk}}\right)e^{\iota2\pi\frac{x_{mk}}{\lambda_c}},
\end{equation} 
with $\lambda_c$ denoting the carrier wavelength, $\iota\triangleq\sqrt{-1}$, $G_m$ and $G_k$ being the gains associated with the antennas of the $m$th AP and $k$th UE respectively, and $\mathbf{a}(\theta_{mk})$ representing the AoA vector at the $m$th AP from the $k$th UE given as,
\begin{equation}
\label{eq:aoavec}
\mathbf{a}(\theta_{mk})=\left[1,e^{\iota2\pi\frac{d}{\lambda_c}\sin(\theta_{mk})},e^{\iota4\pi\frac{d}{\lambda_c}\sin(\theta_{mk})},\ldots,e^{\iota2(N-1)\pi\frac{d}{\lambda_c}\sin(\theta_{mk})}\right]^T.
\end{equation}
\subsection{LoS link probability model}
We note that the channel between the $m$-th AP and the $k$-th UE can either be LoS or NLoS depending on the existence of blockages between them, and the heights of these blockages. Therefore, the LoS link probability $P_{mk}$ depends on geographical distribution parameters of the network, such as the building height distribution, building locations, heights of APs and UEs, etc.. The existing 3GPP model~\cite{3GPPr112_2} however does not account for all these factors, and hence cannot characterize the exact LoS link probabilities. A more complete model for the LoS link probabilities has been discussed in~\cite{Kim_TWC_2020} using the ITU blockage model~\cite{IMT2020propagate}. Using this model, the LoS link probability, $P_{mk}$,  is given by
\begin{equation}
	\label{eq:losprob}
	P_{mk} =  (1 - \omega)^{\sqrt{\alpha \, \mu \, d_{mk}}} 
\end{equation}

\noindent where $\omega \triangleq \sqrt{\frac{\pi}{2}}\frac{\gamma}{\ell_m – \ell^{\prime}_k} \left[ \text{erf}\left(\frac{\ell_m}{\gamma\sqrt{2}}\right) - \text{erf}\left(\frac{\ell^{\prime}_k}{\gamma\sqrt{2}}\right)\right]$, $\text{erf}(z) \triangleq \frac{1}{\sqrt{\pi}}\int_{-z}^{z}e^{-t^2} dt$,
%Recently, different mathematical models have been conceived to characterize $P_{mk}$. For instance, in 3GPP urban channel environment model for micro-cells, the effects of antenna heights at the UE and AP are ignored, $P_{mk}$ is given by~\cite{Fem_Access_2020,Mondal_CMAG_2015,3GPPr112_2}
%%
%%
%\begin{equation}
%P_{mk}=\min\left(1,\frac{d_0}{d_{mk}} +\left(1-\frac{d_0}{d_{mk}} \right)\ e^{-\frac{d_{mk}}{d_0}} \right)\, ,
%\label{eq:2D}
%\end{equation}
%%
%%
%An even more accurate model for LoS probability, accounting for antenna heights at APs and UEs, as well as the height and location distributions of blockages in the network, has been suggested in~\cite{IMT2020propagate}, with $P_{mk}$ given as
%%
%%
%\begin{equation}
%P_{mk}=\prod_{l=0}^{L}1-e^{\frac{\max(\ell_m,\ell_{k}^{'})-\frac{(l+0.5)(\ell_m-\ell_{k}^{'})}{L+1}}{2\gamma^2}}.
%\label{eq:3D}
%\end{equation}
%
%
$\gamma$ is the average altitude of blockages,  $\alpha$ is the fraction of the built up area, and $\mu$ is the average number of blockages per unit area~\cite{IMT2020propagate}.

\subsection{Uplink Channel Estimation}
For simplicity we assume the availability of $K$ orthogonal pilot sequences in each coherence interval, with the pilot symbol sent by the $k$th UE at the $n$th instant given as $\psi_k[n]$, such that $\sum\limits_{n=1}^{K}\psi_k[n]\psi_{l}^*[n]=\delta[k-l]$, {where $\delta[.]$ denotes the Kroneckar delta}. We also assume that all the users share the same average pilot power, $\mathcal{E}_p$. Therefore, the pilot signal received by the $m$th AP at the $n$th instant is given by
\begin{equation}
\mathbf{y}_m[n]=\sum_{k=1}^K \sqrt{\mathcal{E}_p}\mathbf{h}_{mk}\psi_{k}[n]+\sqrt{N_0}\mathbf{w}_m[n].
\end{equation}

\noindent where $\mathbf{w}_m[n] \sim \mathcal{CN}(\mathbf{0}_N,\textbf I_N)$ denotes the independent additive white Gaussian noise (AWGN). Defining $\mathbf{y}^{'}_{mk}\triangleq\sum\limits_{n=1}^{K}\mathbf{y}_{m}[n]\psi_{l}^*[n]$, we obtain
\begin{equation}
\mathbf{y}^{'}_{mk}= \sqrt{\mathcal{E}_p}\mathbf{h}_{mk}+\sqrt{N_0}\mathbf{w}^{'}_k[n],
\label{eq:processpilot}
\end{equation}
with $\mathbf{w}^{'}_k[n]$ defined implicitly. 
Now, given $\beta_{mk}$, and the knowledge of $\delta_{mk}$ at the APs, we have
\begin{equation}
\bm{\Sigma}_{hh,mk}\triangleq\E{\mathbf{h}_{mk}\mathbf{h}^H_{mk}\vert \delta_{mk}}=\delta_{mk}G_mG_k\left( \frac{\ell'_k\ell_m}{4\pi x_{mk}} \right)^2\mathbf{a}(\theta_{mk})\mathbf{a}^{H}(\theta_{mk})+\beta_{mk}\mathbf{I}_N.
\label{eq:sigma_hh}
\end{equation}
Consequently we can write the covrinace matrix of $\mathbf{y}'_{mk}$ as,
\begin{equation}
\bm{\Sigma}_{y^{'}y^{'},mk}\triangleq\E{\mathbf{y}^{'}_{mk}\mathbf{y}^{'H}_{mk}\vert \delta_{mk}}=\mathcal{E}_p\bm{\Sigma}_{hh,mk}+N_0\mathbf{I}_N,
\label{eq:sigma_yy}
\end{equation}
and the cross covariance matrix between $\mathbf{y}'_{mk}$, and $\mathbf{h}_{mk}$ as,
\begin{equation}
\bm{\Sigma}_{hy',mk}\triangleq\E{\mathbf{h}_{mk}\mathbf{y}^{'H}_{mk}\vert \delta_{mk}}=\sqrt{\mathcal{E}_p}\bm{\Sigma}_{hh,mk}.
\label{eq:sigma_hy}
\end{equation}
We can now write the LMMSE estimate $\mathbf{\hat{h}}_{mk}$ of $\mathbf{h}_{mk}$ as~\cite{Kay1}
\begin{equation}
\mathbf{\hat{h}}_{mk}=\sqrt{\mathcal{E}_p}\bm{\Sigma}_{hh,mk}\bm{\Sigma}^{-1}_{y'y',mk}\mathbf{y}^{'}_{mk}.
\end{equation}
Letting $\mathbf{C}_{mk}\triangleq \E{\mathbf{\hat{h}}_{mk}\mathbf{\hat{h}}^H_{mk}}= \mathcal{E}_p\bm{\Sigma}_{hh,mk}\bm{\Sigma}^{-1}_{y'y',mk}\bm{\Sigma}_{hh,mk}$, and $\bar{\mathbf{C}}_{mk}=\bm{\Sigma}_{hh}-\mathbf{C}_{mk}$ we can now express $\mathbf{h}_{mk}$ as
\begin{equation}
\mathbf{{h}}_{mk}=\mathbf{C}_{mk}^{\frac{1}{2}}\mathbf{\hat{h}}_{mk}+\mathbf{\bar{C}}_{mk}^{\frac{1}{2}}\mathbf{\tilde{h}}_{mk},
\label{eq:ch_est}
\end{equation}
where $\mathbf{\tilde{h}}_{mk}$ is a ZMCSCG random vector with i.i.d. unit variance entries, and $E[\mathbf{\tilde{h}}_{mk}\mathbf{\hat{h}}^H_{mk}]=\mathbf{O}_N$, where $\mathbf{O}_N$ corresponds to the all zero order $N$ square matrix.

%Here, we note that it is also common to assume the exact knowledge of the deterministic LoS component of $\mathbf{h}_{mk}$, and consider only the estimate of $\mathbf{\dot{h}}_{mk}$~\cite{Li_TCOM_2016}. In this paper, however, we first evaluate the system performance under ideal CSI scenario, and then with the LMMSE estimate.

\section{Uplink Performance Analysis with Conjugate Beamforming}\label{sec:stream_wise}
In this section, we evaluate the uplink performance of a CF massive MIMO system under LoS/ NLoS channels with conjugate beamformaing and stream wise data detection at the CPU. We assume that all the UEs simultaneously transmit data to the APs over the same time frequency resource. The APs, upon receiving the symbols, perform conjugate beamforming~\cite{Ngo_TWC_2017} and forward the combined symbols to the CPU for stream-wise data detection. That is, for ease of computation, the data transmitted by each user is detected separately. % To evaluate the performance of this scheme, we derive bounds on the achievable rates, and the union bound on the symbol error rate for finite constellations.
\subsection{Achievable Rate Analysis}
Assuming that the $k$th UE transmits the symbol $s_k[n]$, such that $\E{\lvert s_k[n] \rvert^2}=\mathcal{E}_{s,k}$, we can write the signal received at the $m$th AP as  
\begin{equation}
\mathbf{y}_{m}=\sum_{k=1}^K\mathbf{h}_{mk}s_k+\sqrt{N_0}\mathbf{w}_m=\mathbf{H}_m\mathbf{s}+\sqrt{N_0}\mathbf{w}_m,
\end{equation}
where $\mathbf{H}_{m}\in\mathbb{C}^{N\times K }=[\mathbf{h}_{m1},\mathbf{h}_{m2},\ldots,\mathbf{h}_{mK}]$, $\mathbf{s}=[s_1,\ldots,s_K]^T$, and $\mathbf{w}_m\sim\mathcal{CN}(\mathbf{0}_N,\mathbf{I}_N)$. {  Using~\eqref{eq:channelcombo}, we can now express $\mathbf{H}_m$ as
$\mathbf{H}_m=\mathbf{\bar{H}}_m\text{diag}(\boldsymbol{\delta}_m)+\mathbf{\dot{H}}_m\text{diag}(\boldsymbol{\beta}_m),
$
with $\boldsymbol{\delta}_m=[\delta_{m1},\delta_{m2},\ldots,\delta_{mK}]^T$, $\boldsymbol{\beta}_m=[\beta_{m1},\beta_{m2},\ldots,\beta_{mK}]^T$,  $\mathbf{\dot{H}}_m=[\mathbf{\dot{h}}_{m1},\mathbf{\dot{h}}_{m2},\ldots,\mathbf{\dot{h}}_{mK}]$, and $\mathbf{\bar{H}}_m=[\mathbf{\bar{h}}_{m1},\mathbf{\bar{h}}_{m2},\ldots,\mathbf{\bar{h}}_{mK}]$. }

Following this, the APs use conjugate beamforming with the available CSI to obtain the signal $r_{mk}$, such that, when accurate CSI is available at the APs, 
\begin{equation}
r_{mk}= \mathbf{h}_{mk}^H \mathbf{y}_m = \mathbf{h}_{mk}^H\sum_{l=1}^{K}\mathbf{h}_{ml}s_l+\sqrt{N}_0\mathbf{h}_{mk}^H\mathbf{w}_m =\sum_{l=1}^{K}g_{m,kl}s_l+z_{mk},
\end{equation}
where $g_{m,kl}$ and $z_{mk}$ are defined implicitly.
Following this, the APs share their respective $r_{mk}$s with the CPU over the backhaul link, wherein these are combined to obtain $r_k=\sum\limits_{m=1}^Mr_{mk}$, that is, %used to detect the transmitted symbol $s_k$.
%First considering accurate CSI at the APs, we can write $r_{k}$ as
\begin{equation}
r_k=g_{kk}s_k+\sum_{\substack{l=1\\l\neq k}}^Kg_{kl}s_l+z_k. 
\end{equation}

Now, considering that the data transmitted by each user is detected individually, and treating the interference due to the other users' data streams as noise~\cite{HandH}, the rate achievable by the $k$th user can be expressed as,
\begin{equation}
R_k=\E{\log_2\left(1+\frac{\lvert g_{kk}\rvert^2\mathcal{E}_{s,k}}{\sum\limits_{\substack{l=1\\l\neq k}}^K\E{|g_{kl}|^2}\mathcal{E}_{s,l}+\text{var}(z_k)}\right)}. 
\label{eq:Rk}
\end{equation}
\begin{lemma}\normalfont
	The rate achievable by the $k$th user under accurate CSI at the APs can be upper bounded as
	\begin{equation}
	R_k\le\log_2\left(1+\frac{\E{\lvert g_{kk}\rvert^2}\mathcal{E}_{s,k}}{\sum\limits_{\substack{l=1\\l\neq k}}^K\E{|g_{kl}|^2}\mathcal{E}_{s,l}+\text{var}(z_k)}\right). 
	\label{eq:upper_Rk}
	\end{equation}
where,

\begin{multline}
\E{\lvert g_{kk} \rvert^2} =N^2G_k^2\left(\frac{\ell_{k}^{'}}{4\pi}\right)^4\left(\sum_{m=1}^{M}P_{mk}G_m^2\left(\frac{\ell_m}{x_{mk}}\right)^4+\sum_{m=1}^{M}\sum_{\substack{l=1\\l\neq m}}^{M}P_{mk}P_{lk}G_mG_l\left(\frac{\ell_m\ell_l}{x_{mk}x_{lk}}\right)^2\right)\\+N\sum_{m=1}^M\beta_{mk}^2+4N\sum_{m=1}^MP_{mk}\beta_{mk}G_kG_m\left(\frac{\ell_{k}^{'}\ell_m}{4\pi x_{mk}}\right)^2,
\label{eq:meangkk}
\end{multline}
\begin{multline}
\E{|g_{kl}|^2}=\sum_{m=1}^MP_{ml}  P_{mk} G_m\sqrt{G_kG_l}\left(\frac{\ell_{k}^{'}\ell_l\ell_m^2}{16\pi^2 x_{mk}x_{ml}}\right)e^{\iota\frac{2\pi}{\lambda_c}(x_{mk}-x_{ml})}\\ \times \sum_{i=1}^N e^{\iota \frac{2\pi d}{\lambda_c}i(\sin(\theta_{mk})-sin(\theta_{ml}))}+N\sum_{m=1}^M\beta_{mk}\beta_{ml}+4NG_k\left(\frac{\ell_{k}^{'}}{4\pi}\right)^2\sum_{m=1}^MP_{mk}\beta_{ml}G_m\left(\frac{\ell_m}{x_{mk} }\right)^2,
\label{eq:vargkl}
\end{multline}
and 

\begin{equation}
\text{var}({z_k})=N_0\mathbf{h}_k^H\mathbf{I}_{MN}\mathbf{h}_k=N_0N\sum_{m=1}^M\left(P_{mk}G_kG_m\left(\frac{\ell_{k}^{'}\ell_m}{4\pi x_{mk}}\right)^2+\beta_{mk}\right).
\end{equation}
\end{lemma}
	\begin{proof}
	Applying Jensen's inequality~\cite{Cover_Thomas_IT}	to~\eqref{eq:Rk}, we obtain~\eqref{eq:upper_Rk}. The statistics of $g_{kk}$ and $g_{kl}$ under ideal CSI are derived in Appendix~\ref{App:A}, and the variance of $z_k$ is derived in Appendix~\ref{App:B}.
	\end{proof}
\begin{theorem}\normalfont
	A lower bound on the rate achievable by the $k$th UE can be written as,
	\begin{equation}
	R_k\ge   \left [ \log_2 \left({\Eb\left[\lvert g_{kk}\rvert^2\right]\mathcal{E}_{s,k}}\right)-\log_2(e)\frac{\text{var}(|g_{kk}|^2)}{2\Eb^2[|g_{kk}|^2]}-\log_2\left({\sum\limits_{\substack{l=1\\l\neq k}}^K\Eb\left[|g_{kl}|^2\right]\mathcal{E}_{s,l}+\text{var}(z_k)}\right) \right]^+,
	\label{eq:lower_Rk}
	\end{equation}
where $[x]^+$ denotes the operation $\max(x,0)$.
\end{theorem}
\begin{proof}
We can write,
\begin{multline}
R_k\ge \left[\E{\log_2\left(\frac{\lvert g_{kk}\rvert^2\mathcal{E}_{s,k}}{\sum\limits_{\substack{l=1\\l\neq k}}^K\E{|g_{kl}|^2}\mathcal{E}_{s,l}+\text{var}(z_k)}\right)}\right]^+\\=\left[\E{\log_2\left({\lvert g_{kk}\rvert^2\mathcal{E}_{s,k}}\right)-\log_2\left({\sum\limits_{\substack{l=1\\l\neq k}}^K\E{|g_{kl}|^2}\mathcal{E}_{s,l}+\text{var}(z_k)}\right)}\right]^+\\ \ge \left[\Eb\left[\log_2\left({\lvert g_{kk}\rvert^2\mathcal{E}_{s,k}}\right)\right]-\log_2\left({\sum\limits_{\substack{l=1\\l\neq k}}^K\Eb\left[|g_{kl}|^2\right]\mathcal{E}_{s,l}+\text{var}(z_k)}\right)\right]^+. 
\label{eq:lower_Rk1}
\end{multline} 
Defining $\psi_k\triangleq\frac{|g_{kk}|^2}{\Eb[|g_{kk}|^2]}$, we get
\begin{equation}
\Eb\left[\log_2\left({\lvert g_{kk}\rvert^2\mathcal{E}_{s,k}}\right)\right]=\Eb\left[\log_2\left({\Eb[\lvert g_{kk}\rvert^2]\mathcal{E}_{s,k}}\psi_k\right)\right]=\log_2\left({\Eb[\lvert g_{kk}\rvert^2]\mathcal{E}_{s,k}}\right)+\Eb\left[\log_2\left(\psi_k\right)\right].
\end{equation}
Since, 
$$\log_2(\psi_k) \ge \log_2(e) \left((\psi_k-1)-\frac{(\psi_k-1)^2}{2}\right),$$ and 
$\Eb[\psi_k]=1$, $\Eb[(\psi_k-1)^2]=\text{var}(\psi_k)$, we obtain~\eqref{eq:lower_Rk}. The variance of $\lvert g_{kk} \rvert^2$ is given in Appendix~\ref{App:A}.
%\begin{equation}
%R_k\ge   \log_2\left({E\left[\lvert g_{kk}\rvert^2\right]\mathcal{E}_{s,k}}\right)-\log_2(e)\frac{\text{var}(|g_{kk}|^2)}{2E^2[|g_{kk}|^2]}-\log_2\left({\sum\limits_{\substack{l=1\\l\neq k}}^KE\left[|g_{kl}|^2\right]\mathcal{E}_{s,l}+\text{var}(z_k)}\right). 
%\end{equation}
\end{proof}
{Similarly, under estimated CSI, the combined signal at the $m$th AP, denoted as $\hat{r}_{mk}$ becomes,}
{\begin{multline}
\hat{r}_{mk}=\mathbf{\hat{h}}_{mk}^H \mathbf{y}_m=\mathbf{\hat{h}}_{mk}^H\sum_{l=1}^{K}\mathbf{h}_{ml}s_l+\sqrt{N}_0\mathbf{\hat{h}}_{mk}^H\mathbf{w}_m\\=\mathbf{\hat{h}}_{mk}^H\mathbf{C}_{mk}^{\frac{1}{2}}\mathbf{\hat{h}}_{mk}+\mathbf{\hat{h}}_{mk}^H\mathbf{\bar{C}}_{mk}^{\frac{1}{2}}\mathbf{\tilde{h}}_{mk}+\mathbf{\hat{h}}_{mk}^H\sum_{l=1}^{K}\mathbf{h}_{ml}s_l+\sqrt{N}_0\mathbf{\hat{h}}_{mk}^H\mathbf{w}_m \\=\hat{g}_{m,kk}s_k+\tilde{g}_{m,kk}s_k+\sum_{\substack{l=1\\l\neq k}}^{K}g_{m,kl}s_l+z_{mk},
\end{multline}}
{with all the terms again defined implicitly.}
Based on this, the accumulated signal at the CPU can be written as,
{\begin{equation}
\hat{r}_k=\sum_{m=1}^{M}\hat{r}_{mk}=\hat{g}_{kk}s_k+\tilde{g}_{kk}s_k+\sum_{\substack{l=1\\l\neq k}}^Kg_{kl}s_l+z_k. 
\label{eq:sig_est_CSI}
\end{equation}}

From~\eqref{eq:sig_est_CSI}, treating interference as noise, we obtain the achievable rate for the $k$th UE as
\begin{equation}
R_k=\Eb\left[\log_2\left(1+\frac{\lvert \hat{g}_{kk}\rvert^2\mathcal{E}_{s,k}}{ \E{\lvert\tilde{g}_{kk}\rvert^2}\mathcal{E}_{s,k}+\sum\limits_{\substack{l=1\\l\neq k}}^K\E{|g_{kl}|^2}\mathcal{E}_{s,l}+\text{var}(z_k)}\right)\right]. 
\label{eq:Rk_est}
\end{equation}

\begin{lemma}\normalfont
	The rate achievable by the $k$th user under estimated CSI at the APs can be upper bounded as
	\begin{equation}
	R_k\le\log_2\left(1+\frac{\Eb\left[\lvert \hat{g}_{kk}\rvert^2\right]\mathcal{E}_{s,k}}{\Eb\left[\lvert \tilde{g}_{kk}\rvert^2\right]\mathcal{E}_{s,k}+\sum\limits_{\substack{l=1\\l\neq k}}^K\Eb\left[|g_{kl}|^2\right]\mathcal{E}_{s,l}+\text{var}(z_k)}\right). 
	\label{eq:upper_Rk_est}
	\end{equation}
where the statistics of $\hat{g}_{kk}$ and $\tilde{g}_{kk}$, are derived in Appendix~\ref{App:C}.
%	\begin{multline}
%	E\left[\lvert \hat{g}_{kk} \rvert^2\right] =N^2G_k^2\left(\frac{\ell_{k}^{'}}{4\pi}\right)^4\left(\sum_{m=1}^{M}P_{mk}c^2_{mk}G_m^2\left(\frac{\ell_m}{x_{mk}}\right)^4+\sum_{m=1}^{M}\sum_{\substack{l=1\\l\neq m}}^{M}P_{mk}P_{lk}c_{mk}c_{lk}G_mG_l\left(\frac{\ell_m\ell_l}{x_{mk}x_{lk}}\right)^2\right)\\+N\sum_{m=1}^Mc^2_{mk}\beta_{mk}^2+4N\sum_{m=1}^Mc^2_{mk}P_{mk}\beta_{mk}G_kG_m\left(\frac{\ell_{k}^{'}\ell_m}{4\pi x_{mk}}\right)^2,
%	\label{eq:Eghat_kk2}
%	\end{multline}
%	\begin{equation}
%	E[\lvert\tilde{g}_{kk}\rvert^2]=N\sum_{m=1}^M\bar{c}_{mk}^2\E{\lVert \mathbf{h}_{mk} \rVert_2^2}\beta^2_{mk}+4NG_k\left(\frac{\ell_{k}^{'}}{4\pi}\right)^2\sum_{m=1}^MP_{mk}\bar{c}^2_{mk}\E{\lVert \mathbf{h}_{mk} \rVert_2^2}\beta_{mk}G_m\left(\frac{\ell_m}{x_{mk} }\right)^2.
%	\label{eq:g_tilde2}
%	\end{equation}
	
\end{lemma}
\begin{proof}
	Applying Jensen's inequality to~\eqref{eq:Rk_est} results in~\eqref{eq:upper_Rk_est}.  
\end{proof}

\begin{theorem}\normalfont
The rate achievable by the $k$th UE in the uplink can be given as,
	\begin{equation}
	R_k^{\text{est}}\ge  \left[ \log_2\left({\Eb\left[\lvert \hat{g}_{kk}\rvert^2\right]\mathcal{E}_{s,k}}\right)-\log_2(e)\frac{\text{var}(|\hat{g}_{kk}|^2)}{2\Eb^2[|\hat{g}_{kk}|^2]}-\log_2\left(\Eb\left[\lvert \tilde{g}_{kk}\rvert^2\right]\mathcal{E}_{s,k}+{\sum\limits_{\substack{l=1\\l\neq k}}^K\Eb\left[|g_{kl}|^2\right]\mathcal{E}_{s,l}+\text{var}(z_k)}\right) \right]^+,
	\label{eq:lower_Rk_est}
	\end{equation}
\end{theorem}
\begin{proof}
	We can write,
	\begin{multline}
	R_k\ge \left[\Eb\left[\log_2\left(\frac{\lvert \hat{g}_{kk}\rvert^2\mathcal{E}_{s,k}}{\Eb\left[\lvert \tilde{g}_{kk}\rvert^2\right]\mathcal{E}_{s,k}+\sum\limits_{\substack{l=1\\l\neq k}}^K\E{|g_{kl}|^2}\mathcal{E}_{s,l}+\text{var}(z_k)}\right)\right]\right]^+
	\\
	=\left[\Eb\left[\log_2\left({\lvert \hat{g}_{kk}\rvert^2\mathcal{E}_{s,k}}\right)-\log_2\left(\Eb\left[\lvert \tilde{g}_{kk}\rvert^2\right]\mathcal{E}_{s,k}+{\sum\limits_{\substack{l=1\\l\neq k}}^K\E{|g_{kl}|^2}\mathcal{E}_{s,l}+\text{var}(z_k)}\right)\right] \right]^+
	\\
	 \ge \left[ \Eb\left[\log_2\left({\lvert \hat{g}_{kk}\rvert^2\mathcal{E}_{s,k}}\right)\right]-\log_2\left(\Eb\left[\lvert \tilde{g}_{kk}\rvert^2\right]\mathcal{E}_{s,k}+{\sum\limits_{\substack{l=1\\l\neq k}}^K\Eb\left[|g_{kl}|^2\right]\mathcal{E}_{s,l}+\text{var}(z_k)}\right)\right]^+. 
	\label{eq:lower_Rk1_est}
	\end{multline} 
	We prove this by defining $\hat{\psi}_k\triangleq\frac{|\hat{g}_{kk}|^2}{\Eb[|\hat{g}_{kk}|^2]}$, and following a procedure similar to Theorem~1.
\end{proof}
\subsection{Computational Complexity Analysis}
With the users transmitting from the finite constellation $\mathcal{S}$, and under the availability of accurate CSI at the CPU, an estimate $\hat{s}_k$ of the symbol $s_k$ transmitted by the $k$th UE can be obtained as,
\begin{equation}
\hat{s}_k=\arg\min_{s\in\mathcal{S}}|r_k-g_{kk}s|.
\end{equation}
For soft symbol decoding, the probability of the $i$th constellation symbol $s_{k,i}\in\mathcal{S}$ being transmitted by the $k$th UE is given as,
\begin{equation}
\Pr\{s_{k,i}|r_k\}=\frac{1}{\pi N_0}\text{exp}\left(-\frac{|r_k-g_{kk}s_{k,i}|^2}{\sum\limits_{\substack{l=1\\l\neq k}}^K\Eb\left[|g_{kl}|^2\right]\mathcal{E}_{s,l}+\text{var}(z_k)}\right).
\end{equation} 

{Similarly, in case of estimated CSI, the hard estimate, and the soft symbol probability can be respectively written as,}
\begin{equation}
\hat{s}_k=\arg\min_{s\in\mathcal{S}}|r_k-\hat{g}_{kk}s|,
\end{equation}
and
{\begin{equation}
\text{Pr}^{\text{est}}\{s_{k,i}|r_k\}=\frac{1}{\pi N_0}\text{exp}\left(-\frac{|r_k-\hat{g}_{kk}s_{k,i}|^2}{\E{|\tilde{g}_{kk}|^2}\mathcal{E}_{s,k}+\sum\limits_{\substack{l=1\\l\neq k}}^K\Eb\left[|g_{kl}|^2\right]\mathcal{E}_{s,l}+\text{var}(z_k)}\right).
\end{equation}}

It can be observed that for $K$ users, and a constellation size $|\mathcal{S}|$, the computational complexity for hard symbol decoding is $\mathcal{O}(K|\mathcal{S}|)$. 

%Consequently, the probability of error for the $k$th UE's data stream is given as,
%\begin{equation}
%P^{\text{est}}_{e,k}\le \sum_{s_{k,i}\in\mathcal{S}}\Pr\{s_{i,k}\}\sum_{s_{j,k}\in\mathcal{S}\backslash s_{k,i}} \Pr\{s_{k,i}\to s_{k,j} \}, 
%\end{equation}
%such that with hard decoding and accurate CSI at the APs, the pairwise probability of symbol error $\Pr\{s_{k,i}\to s_{k,j} \}$, is given as,
%\begin{equation}
%\Pr\{s_{k,i}\to s_{k,j} \}=E_{g_{kk}}\left[Q\left(\frac{g_{kk}|(s_{k,i}- s_{k,j})|}{\sqrt{\sum\limits_{\substack{l=1\\l\neq k}}^KE\left[|g_{kl}|^2\right]\mathcal{E}_{s,l}+\text{var}(z_k)}}\right)\right]. 
%\end{equation}
%and under estimated CSI, it is given as,
%\begin{equation}
%\text{Pr}^{\text{est}}\{s_{k,i}\to s_{k,j} \}=E_{\hat{g}_{kk}}\left[Q\left(\frac{\hat{g}_{kk}|(s_{k,i}- s_{k,j})|}{\sqrt{|\tilde{g}_{kk}|^2\mathcal{E}_{s,k}+\sum\limits_{\substack{l=1\\l\neq k}}^KE\left[|g_{kl}|^2\right]\mathcal{E}_{s,l}+\text{var}(z_k)}}\right)\right]. 
%\end{equation}

\section{Performance Analysis with Joint Detection}\label{sec:joint}
In case of joint detection, we assume that CSI for all the channels is available with the CPU. In this case, the APs forward their received signals to the CPU, so that the concatenated symbol at the CPU is given by
\begin{equation}
\mathbf{y}=\mathbf{H}\mathbf{s}+\sqrt{N_0}\mathbf{w},
\label{eq:sig_central}
\end{equation}
where $\mathbf{H}=[\mathbf{H}_1^T,\mathbf{H}_2^T,\ldots,\mathbf{H}_M^T]^T$, and $\mathbf{w}=[\mathbf{w}_1^T,\mathbf{w}_2^T,\ldots,\mathbf{w}_M^T]^T$.
Therefore, the achievable sum rate for joint detection becomes, 
\begin{equation}
R=\Eb_{\mathbf{H}}\left[\log_2\left(\det \left( \mathbf{I}_{MN}+\mathbf{H}\text{diag}(\bm{\mathcal{E}}_s/N_0)\mathbf{H}^H \right) \right)\right].
\label{eq:R_joint}
\end{equation}
\begin{lemma}\normalfont
The sum rate for the CF massive MIMO system with LoS/NLoS channels under the availability of accurate CSI, and joint detection at the CPU, is upper bounded as 
\begin{equation}
R\le\log_2\left(\det \left( \mathbf{I}_{MN}+\text{diag}(\bm{\mathcal{E}}_s/N_0)\Eb\left[\mathbf{G}\right] \right) \right),
\label{eq:T5}
\end{equation}
with the $k$th diagonal entry of the matrix $\Eb[\mathbf{G}]$ being
\begin{equation}
\Eb[g_{kk}]=N\sum_{m=1}^M\left(P_{mk}G_kG_m\left(\frac{\ell_{k}^{'}\ell_m}{4\pi x_{mk}}\right)^2+\beta_{mk}\right), 
\end{equation}
and its $(k,l)$th off diagonal entry being,
\begin{equation}
\Eb[g_{kl}]=\sum_{m=1}^MP_{ml}P_{mk}G_m\sqrt{G_kG_l}\left(\frac{\ell_{k}^{'}\ell_l\ell_m^2}{16\pi^2 x_{mk}x_{ml}}\right)e^{\iota\frac{2\pi}{\lambda_c}(x_{mk}-x_{ml})}\sum_{i=1}^N e^{\iota \frac{2\pi d}{\lambda_c}i(\sin(\theta_{mk})-sin(\theta_{ml}))}.
\end{equation}
\end{lemma}
\begin{proof}
	Since $\det(\mathbf{I}+\mathbf{AB})=\det(\mathbf{I}+\mathbf{BA})$~\cite{Horn_MA}, we obtain,
	\begin{equation}
	R=\Eb_{\mathbf{H}}\left[\log_2\left(\det \left( \mathbf{I}_{MN}+\text{diag}(\bm{\mathcal{E}}_s/N_0)\mathbf{H}^H\mathbf{H} \right) \right)\right].
	\label{eq:upper_T5}
	\end{equation}
	Applying the Jensen's inequality to~\eqref{eq:upper_T5}, we obtain~\eqref{eq:T5}, and the statistics of $\mathbf{G}$ are derived in Appendix~\ref{App:A}.
\end{proof}
\begin{theorem}\label{thm:perfect_lower_joint}
\normalfont
The sum rate for the CF massive MIMO system with LoS/NLoS channels under the availability of accurate CSI, and joint detection at the CPU, is lower bounded as
\begin{equation}
R\ge \left[\log_2\left(\det \left( \text{diag}(\bm{\mathcal{E}}_s/N_0)\Eb[\mathbf{G}]\right) \right) +\Eb\left[\log_2(\det(\bm{\Psi})) \right] \right]^+,
\label{eq:T6}
\end{equation} 
where $\bm{\Psi}=(\Eb[\mathbf{G}])^{-1}\mathbf{G}$. Here, $\Eb\left[\log_2(\det(\bm{\Psi})) \right]$, needs to be computed empirically.
\end{theorem}
\begin{proof}
From~\eqref{eq:R_joint}, it is easy to show that
\begin{equation}
R\ge \left[\Eb\left[\log_2\left(\det \left( \text{diag}(\bm{\mathcal{E}}_s/N_0)\mathbf{H}^H\mathbf{H} \right) \right)\right]\right]^+,
\end{equation}
Letting $\mathbf{G}=\mathbf{H}^H\mathbf{H}$, and $\bm{\Psi}=(\Eb[\mathbf{G}])^{-1}\mathbf{G}$, we have
\begin{equation}
R\ge \left[ \Eb \left[\log_2\left(\det \left( \text{diag}(\bm{\mathcal{E}}_s/N_0)\Eb[\mathbf{G}]\right) +\log_2(\det(\bm{\Psi}))  \right)\right]\right]^+,
\end{equation}
%based on this, we can now write,
%\begin{equation}
%R\ge \left[\Eb\left[\log_2\left(\det \left( \text{diag}(\bm{\mathcal{E}}_s/N_0)\Eb[\mathbf{G}]\right) \right) +\log_2(\det(\bm{\Psi})) \right]\right]^+,
%\end{equation}
this can be trivially reduced to~\eqref{eq:T6}. 

We note that the entries of $\bm{\Psi}$ come from different distributions, and it is not possible obtain the distribution of $\log_2(\det(\bm{\Psi}))$ in a closed form. This term therefore needs to be evaluated empirically using Monte Carlo simulations.
\end{proof}
In case the value of $\log_2(\det(\bm{\Psi}))$  is small, we can approximate the achievable sum rate as, 
\begin{equation}
R\approx \left[ \log_2\left(\det \left( \text{diag}(\bm{\mathcal{E}}_s/N_0)\Eb[\mathbf{G}]\right) \right) \right]^+.
\end{equation}
\noindent
Now, if the estimated CSI is available at the CPU, the concatenated received signal in~\eqref{eq:sig_central} can be rewritten using~\eqref{eq:ch_est} as follows
\begin{equation}
\mathbf{y}=\mathbf{\hat{H}}\mathbf{s}+\mathbf{\tilde{H}}\mathbf{s}+\sqrt{N}_0\mathbf{w}.
\label{eq:sig_cent_est}
\end{equation}
Here  $\mathbf{\hat{H}}=[\mathbf{\hat{h}}_1,\mathbf{\hat{h}}_2,\ldots,\mathbf{\hat{h}}_K]$, and,  $\mathbf{\tilde{H}}=[\mathbf{\tilde{h}}_1,\mathbf{\tilde{h}}_2,\ldots,\mathbf{\tilde{h}}_K]$, where, \\$\mathbf{\hat{h}}_k=[(\mathbf{C}^{\frac{1}{2}}_{1k}\mathbf{\hat{h}}_{1k})^T,(\mathbf{C}^{\frac{1}{2}}_{2k}\mathbf{\hat{h}}_{2k})^T,\ldots,(\mathbf{C}^{\frac{1}{2}}_{Mk}\mathbf{\hat{h}}_{Mk})^T]^T$, and $\mathbf{\tilde{h}}_k=[(\mathbf{\bar{C}}^{\frac{1}{2}}_{mk}\mathbf{\hat{h}}_{1k})^T,(\mathbf{C}^{\frac{1}{2}}_{mk}\mathbf{\hat{h}}_{2k})^T,\ldots,(\mathbf{C}^{\frac{1}{2}}_{mk}\mathbf{\hat{h}}_{1k})^T]^T$ respectively.
From~\eqref{eq:sig_cent_est}, treating interference as noise, we can write the achievable sum rate for all the users as,
\begin{equation}
R=\Eb_{\mathbf{\hat{H}}}\left[\log_2\left(\det \left( \mathbf{I}_{MN}+\mathbf{\hat{H}}\text{diag}(\bm{\mathcal{E}}_s/\sigma_y^2)\mathbf{\hat{H}}^H \right) \right)\right],
\label{eq:R_joint_est}
\end{equation}
where $\sigma_y^2$ is the total noise and interference power, and is given by,
\begin{equation}
\sigma_y^2=\sum_{k=1}^K\sum_{m=1}^M \mathcal{E}_{s,k} \Eb[\lVert \bar{\mathbf{C}}^{\frac{1}{2}}_{mk}\tilde{\mathbf{h}}_{mk} \rVert^2]+N_0.
\end{equation}
\begin{lemma}\normalfont
	The sum rate for the CF massive MIMO system with  LoS/NLOS channels under estimated CSI and joint data detection is upper bounded as,
	\begin{equation}
	R\le\log_2\left(\det \left( \mathbf{I}_{MN}+ (\bm{\mathcal{E}}_s/\sigma_y^2)\Eb[\mathbf{\hat{G}}] \right) \right),
	\label{eq:upper_R_joint_est}
	\end{equation}
\end{lemma}
where $\mathbf{\hat{G}}=\mathbf{\hat{H}}^H\mathbf{\hat{H}}$, whose statistics are derived in Appendix~\ref{App:C}.
\begin{proof}
	This is a direct consequence of applying the Jensen's inequality to~\eqref{eq:R_joint_est}.
\end{proof}
\begin{lemma}\normalfont
	The sum rate for the CF massive MIMO system with LoS/NLoS channels under estimated CSI and joint data detection is lower bounded as, 
	{\begin{equation}
	R\ge\left[\log_2\left(\det \left(  (\bm{\mathcal{E}}_s/\sigma_y^2)\Eb[\mathbf{\hat{G}}] \right) \right) +\Eb\left[\log_2(\det(\bm{\Psi})) \right]\right]^+.
	\label{eq:lower_R_joint_est}
	\end{equation}}
\end{lemma}
\begin{proof}
	Starting with \eqref{eq:R_joint_est} and following a procedure similar to the proof of Theorem~\ref{thm:perfect_lower_joint}, we obtain~\eqref{eq:lower_R_joint_est}.
\end{proof}

\subsection{Computational Complexity Analysis}~\label{sec:joint_CC}
In case the UEs transmit symbols from a finite constellation $\mathcal{S}$, the overall symbol received by the CPU belongs to the constellation $\mathcal{S}^K$, and the hard estimate of the transmitted symbol vector, $\mathbf{\hat{s}}$, under accurate and estimated CSI can be respectively obtained as,
\begin{equation}
\mathbf{\hat{s}}=\arg\min_{\mathbf{s} \in \mathcal{S}^K}\lVert \mathbf{y}-\mathbf{H}\mathbf{s}\rVert^2,
\end{equation}
and
\begin{equation}
\mathbf{\hat{s}}=\arg\min_{\mathbf{s} \in \mathcal{S}^K}\lVert \mathbf{y}-\mathbf{\hat{H}}\mathbf{s}\rVert^2.
\end{equation}
Similarly, the probabilities of the symbol vector $\mathbf{s}_i\in \mathcal{S}^K$, $i \in \{1,2,\ldots,|\mathcal{S}|^K \}$, in the two cases can be respectively computed as,
\begin{equation}
\Pr\{\mathbf{s}_i|\mathbf{y}\}=\frac{1}{(\pi N_0)^K}\exp\left(- \frac{1}{N_0}\lVert \mathbf{y}-\mathbf{H}\mathbf{s}\rVert^2\right),
\end{equation}
\begin{equation}
\Pr\{\mathbf{s}_i|\mathbf{y}\}=\frac{1}{(\pi \sigma_y^2)^K}\exp\left(- \frac{1}{\sigma_y^2}\lVert \mathbf{y}-\mathbf{\hat{H}}\mathbf{s}\rVert^2\right).
\label{}
\end{equation}
Therefore, the computational complexity in this case becomes $\mathcal{O}(\lvert S \rvert^K)$.
%Also, invoking the union bound, the probability of symbol error can be expressed as,
%\begin{equation}
%{P}_e\le \sum_{\mathbf{s}_i\in\mathcal{S}^K }\Pr\{\mathbf{s}_i\}\sum_{\mathbf{s}_j\in\mathcal{S}^K \backslash \mathbf{s}_i} \Pr\{\mathbf{s}_i\to \mathbf{s}_j\},
%\end{equation}
%with $\Pr\{\mathbf{s}_i\to \mathbf{s}_j\}$ being the pairwise probability of error, and being calculated under accurate and estimated CSI respectively as,
%\begin{equation}
%\Pr\{\mathbf{s}_i\to \mathbf{s}_j\}=E_{\mathbf{H}} \left[ Q\left(\sqrt{\frac{1}{N_0}}\lVert \mathbf{H}\mathbf{s}_i-\mathbf{H}\mathbf{s}_j \rVert\right)\right],
%\end{equation}
%\begin{equation}
%\Pr\{\mathbf{s}_i\to \mathbf{s}_j\}=E_{\mathbf{H}} \left[ Q\left(\sqrt{\frac{1}{N_0}}\lVert \mathbf{\hat{H}}\mathbf{s}_i-\mathbf{\hat{H}}\mathbf{s}_j \rVert\right)\right].
%\end{equation}
%{The exact evaluation of these expression is analytically intractable, and therefore the joint BER for a CF massive MIMO system needs to be evaluated using Monte Carlo simulations.}
%
%
%
\section{Performance under MMSE Data Combining at the CPU}~\label{sec:message}
We observe that in case of conjugate beamfroming and subsequent communication of the combined symbols to the CPU for streamwise decoding, the communication and computation requirements at the CPU are minimum, but the resultant inter-user interference is quite significant. % However, this approach leads to significant inter stream interference among the users. %\red{Furthermore, this interference leads to the ``capture" of the APs by the users with LoS channels, while the users with NLoS channels go in outage.} \textit{This limits the overall achievable sum rate in systems where both LoS and NLoS channels are present. }
On the other hand, using the centralized ML detection based decoding at the CPU results in higher sum rates. This, however, comes at the expense of high communication and computation costs at the CPU. We observe that the decoding complexity increases exponentially with the number of UEs~(see Section~\ref{sec:joint_CC}). Therefore, in this section, we explore the possibility of centralized MMSE based combining of the received samples followed by stream-wise decoding. In this case, the signal received by the CPU is the same as that in case of joint detection, as given in~\eqref{eq:sig_central}.

Now, in case the CPU has access to the accurate CSI, it uses the combining matrix $\mathbf{V}=(\mathbf{H}\mathbf{H}^H+N_0\mathbf{I}_{MN})^{-1}\mathbf{H}$, to generate the signal vector $\mathbf{r}=\mathbf{V}^H\mathbf{y}$, such that,
\begin{equation}
\mathbf{r}=\mathbf{H}^H(\mathbf{H}\mathbf{H}^H+N_0\mathbf{I}_{MN})^{-1}\mathbf{H}\mathbf{s}+\sqrt{N_0}\mathbf{H}^H(\mathbf{H}\mathbf{H}^H+N_0\mathbf{I}_{MN})^{-1}\mathbf{w}.
\end{equation}
Based on this, the combined signal for the data sent over the $k$th stream can be expressed as,
\begin{multline}
r_k=\mathbf{h}_{k}^H(\mathbf{H}\mathbf{H}^H+N_0\mathbf{I}_{MN})^{-1}\mathbf{h}_ks_k+\sum_{\substack{l=1\\l\neq k}}^{K}\mathbf{h}_{k}^H(\mathbf{H}\mathbf{H}^H+N_0\mathbf{I}_{MN})^{-1}\mathbf{h}_ls_l\\+\sqrt{N_0}\mathbf{h}_k^H(\mathbf{H}\mathbf{H}^H+N_0\mathbf{I}_{MN})^{-1}\mathbf{w}=f_{kk}s_k+\sum_{\substack{l=1\\l\neq k}}^{K}f_{kl}s_l+\xi_k,
\label{eq:MMSE_exact}
\end{multline}
with the effective channel coefficients $f_{kl}$, and the effective noise term $\xi_k$ defined implicitly. Based on~\eqref{eq:MMSE_exact}, and considering interference as noise, we can write the rate achievable by the $k$th UE as,
\begin{equation}
R^{\text{MMSE}}_k=\E{\log_2\left(1+\frac{\left\lvert f_{kk}\right\rvert^2\mathcal{E}_{s,k}}{\sum_{\substack{l=1\\l\neq k}}^{K}\E{\left \lvert f_{kl} \right\rvert^2}\mathcal{E}_{s,l}+ \var (\xi_k) }\right) }.
\label{eq:MMSE_comb}
\end{equation}

\begin{lemma}
\normalfont $R^{\text{MMSE}}_k$ can be upper bounded as,
\begin{equation}
R^{\text{MMSE}}_k\le\log_2\left(1+\frac{\E{\left\lvert f_{kk}\right\rvert^2}\mathcal{E}_{s,k}}{\sum_{\substack{l=1\\l\neq k}}^{K}\E{\left \lvert f_{kl} \right\rvert^2}\mathcal{E}_{s,l}+ \var(\xi_k) }\right) .
\end{equation}
\end{lemma}
\begin{proof}
	This is a direct consequence of applying the Jensen's inequality to~\eqref{eq:MMSE_comb}.
\end{proof}
\begin{lemma}\normalfont
	A lower bound on the rate achievable by the $k$th UE, $R^{\text{MMSE}}_k$, can be written as,
	\begin{equation}
	R^{\text{MMSE}}_k\ge  \left[ \log_2\left({\Eb\left[\lvert f_{kk}\rvert^2\right]\mathcal{E}_{s,k}}\right)-\log_2(e)\frac{\text{var}(|f_{kk}|^2)}{2\Eb^2[|f_{kk}|^2]}-\log_2\left({\sum\limits_{\substack{l=1\\l\neq k}}^K\Eb\left[|f_{kl}|^2\right]\mathcal{E}_{s,l}+\text{var}(\xi_k)}\right)\right]^+.
	\label{eq:lower_Rk_MMSE}
	\end{equation}
\end{lemma}
\begin{proof}
	This can be shown using steps similar to those followed in Theorem 2.
\end{proof}
\noindent We note that, it is not possible to calculate the statistics of $f_{kl}$ in a closed form. Hence, these expressions need to be evaluated empirically.

Similarly, when only the estimated CSI is available at the CPU, the combining matrix becomes, 
$\mathbf{V}=\left(\mathbf{\hat{H}}\mathbf{\hat{H}}^H+\psi\mathbf{I}_{MN}\right)^{-1}\mathbf{\hat{H}}$, with $\psi$ being the regularization parameter~(for the generalized MMSE combining)\footnote{It is difficult to obtain the optimal value of $\psi$ in a closed form, and this needs to be optimized empirically.}. Consequently the combined signal for the $k$th user's stream can be expressed as,
\begin{multline}
r_k=\mathbf{\hat{h}}_{k}^H\left(\mathbf{\hat{H}}\mathbf{\hat{H}}^H+\psi\mathbf{I}_{MN}\right)^{-1}\mathbf{\hat{h}}_ks_k+\sum_{\substack{l=1\\l\neq k}}^{K}\mathbf{\hat{h}}_{k}^H\left(\mathbf{\hat{H}}\mathbf{\hat{H}}^H+\psi\mathbf{I}_{MN}\right)^{-1}\mathbf{\hat{h}}_ls_l\\+\sum_{\substack{l=1}}^{K}\mathbf{\hat{h}}_{k}^H(\mathbf{\hat{H}}\mathbf{\hat{H}}^H+\psi\mathbf{I}_{MN})^{-1}\mathbf{\tilde{h}}_ls_l+\sqrt{N_0}\mathbf{\hat{h}}_{k}^H\left(\mathbf{\hat{H}}\mathbf{\hat{H}}^H+\psi\mathbf{I}_{MN}\right)^{-1}\mathbf{w}\\=\hat{f}_{kk}s_k+\sum_{\substack{l=1\\l\neq k}}^{K}\hat{f}_{kl}s_l+\sum_{\substack{l=1}}^{K}\tilde{f}_{kl}s_l+\xi_k,
\end{multline}
with $\hat{f}_{kl}$, $\tilde{f}_{kl}$, and $\xi_k$ defined implicitly.  We can now write the rate achievable by the $k$th UE as,
\begin{equation}
R_k^{\text{MMSE}}=\E{\log_2\left(1+\frac{\left\lvert \hat{f}_{kk}\right\rvert^2\mathcal{E}_{s,k}}{\sum_{\substack{l=1\\l\neq k}}^{K}\E{\left \lvert \hat{f}_{kl} \right\rvert^2}\mathcal{E}_{s,l}+\sum_{\substack{l=1}}^{K}\E{\left \lvert \tilde{f}_{kl} \right\rvert^2}\mathcal{E}_{s,l}+ \var (\xi_k) }\right) }.
\label{eq:MMSE_comb_est}
\end{equation}
 We can again determine the bounds on the rates achievable by the $k$th user, using the results derived previously. However since closed form expressions for the moments of the effective channels coefficients are intractable, we omit the explicit statement of these expressions. 
 
\textsc{Computational Complexity:}
In case of the UEs transmitting symbols from a finite constellation $\mathcal{S}$, with accurate CSI being available at the CPU, 
the estimate of the symbol $s_k$ transmitted by the $k$th UE can be obtained as,
\begin{equation}
\hat{s}_k=\arg\min_{s\in\mathcal{S}}|r_k-f_{kk}s|.
\end{equation}
For soft symbol decoding, the probability of the $i$th constellation symbol, $s_{k,i}\in\mathcal{S}$, being transmitted by the $k$th UE is given as,
\begin{equation}
\Pr\{s_{k,i}|r_k\}=\frac{1}{\pi \left (\sum\limits_{\substack{l=1\\l\neq k}}^K\Eb\left[|g_{kl}|^2\right]\mathcal{E}_{s,l}+\text{var}(\xi_k)\right)}\text{exp}\left(-\frac{|r_k-f_{kk}s_{k,i}|^2}{\sum\limits_{\substack{l=1\\l\neq k}}^K\Eb\left[|g_{kl}|^2\right]\mathcal{E}_{s,l}+\text{var}(\xi_k)}\right).
\end{equation} 

Similarly, in case of estimated CSI, the hard estimate, and the soft symbol probability can be respectively written as,
\begin{equation}
\hat{s}_k=\arg\min_{s\in\mathcal{S}}|r_k-\hat{f}_{kk}s|,
\end{equation}
and
{\begin{equation}
	\text{Pr}\{s_{k,i}|r_k\}=\frac{1}{\pi\sigma_r^2}\text{exp}\left(-\frac{|r_k-\hat{g}_{kk}s_{k,i}|^2}{\sigma_r^2}\right),
	\end{equation}}
where $$\sigma_r^2=\sum_{l=1}^K\E{|\tilde{f}_{kl}|^2}\mathcal{E}_{s,l}+\sum\limits_{\substack{l=1\\l\neq k}}^K\Eb\left[|\hat{f}_{kl}|^2\right]\mathcal{E}_{s,l}+\text{var}(\xi_k).$$

In this case, with $K$ users, $M$ APs, and $N$ antennas per AP, for a constellation size $|\mathcal{S}|$, the computational complexity for hard symbol decoding is $\mathcal{O}((MN)^3+K|\mathcal{S}|)$.

\section{Simulation Results}\label{sec:results}
In this section we present simulation and numerical results to corroborate the validity of the bounds derived by us, and to analyze the performance of the CF massive MIMO system under LoS/ NLoS channels. Unless stated otherwise, the simulation parameters used for these experiments are listed in Table~\ref{tab:sims}.

\begin{table}[!t]
\caption{Simulation Parameters~\cite{IMT2020guide,IMT2020propagate} \label{tab:sims}} 

\begin{center}
\rowcolors{2}%{green!20!yellow}
{cyan!15!}{}
\renewcommand{\arraystretch}{1.0}
\begin{tabular}{l | p{4 cm} | l || l | p{4 cm} | l }
\hline 
 {\bf Parameters} & {\bf Descriptions} & {\bf Values} &  {\bf Parameters} & {\bf Descriptions} & {\bf Values}   \\
\hline 
\hspace{0.15cm}$d_0$ & Reference distance 
			& \hspace{0.12cm}$1$~m  &
			\hspace{0.15cm}$f_c$  & Carrier Frequency 
			& \hspace{0.12cm}$2$ GHz \\ %\addlinespace
\hspace{0.15cm}$\ell_m$  & AP antenna height
						& \hspace{0.12cm}$10$~m   &
\hspace{0.15cm}$\ell_k^{'}$   & UE antenna height
			& \hspace{0.12cm}$1.5$~m   \\ %\addlinespace
\hspace{0.15cm}$K$  & Number of UEs
			& \hspace{0.12cm}$64$   &
			\hspace{0.15cm}$MN$  & AP antenna density
						& \hspace{0.12cm}$1024$/km\textsuperscript{2}   \\ %\addlinespace
\hspace{0.15cm}$\mu$  & Average number of blockages/ unit area
			& \hspace{0.12cm}$300$/km\textsuperscript{2} &
			\hspace{0.15cm}$\alpha$  & Fraction of built-up area in the network
						& \hspace{0.12cm}$0.5$ \\ %\addlinespace
\hspace{0.15cm} $\gamma$    & Average altitude of blockage
			& \hspace{0.12cm}$20$ m  &
\hspace{0.15cm}   & Network area for simulation
			& $1$~km\textsuperscript{2} \hspace{0.12cm}  \\ %\addlinespace
\hline
\end{tabular} %\vspace{-0.3 cm}
\end{center}
\end{table}%
%
%

%%%\begin{table}	
%%%
%%%	\caption{Simulation Parameters~\cite{IMT2020propagate}}
%%%	\centering
%%%	\begin{tabular}{|c|c|}
%%%		\hline \\
%%%		\bf{Parameter}&\bf{Value}\\
%%%		\hline \\
%%%		Reference distance ($d_0$) & 1~m \\
%%%		\hline \\
%%%		Network Area & $1 \text{km}^2$ \\
%%%		\hline
%%%		\\ AP Antenna Density & $1024/ \text{km}^2$
%%%		\\ \hline
%%%		\\ Number of Users & 64 
%%%		\\ \hline
%%%		\\ AP Height & 10m
%%%		\\ \hline
%%%		\\ UE antenna height & 1.5m
%%%		\\ \hline
%%%		\\ Carrier frequency ($f_c$) & 2 GHz 
%%%		\\ \hline
%%%		\\ $\alpha$& 0.5 
%%%		\\ \hline 
%%%		\\ $\mu$& $300/km^2$ 
%%%		\\ \hline 
%%%		\\ Average building height ($\gamma$) & 20~m
%%%	 \\ \hline 
%%%		
%%%		
%%%		
%%%		
%%%		
%%%			\end{tabular}
%%%		\label{tab:sims}
%%%\end{table} 
\begin{figure}[t]
	\centering
	\includegraphics[width=0.7\textwidth]{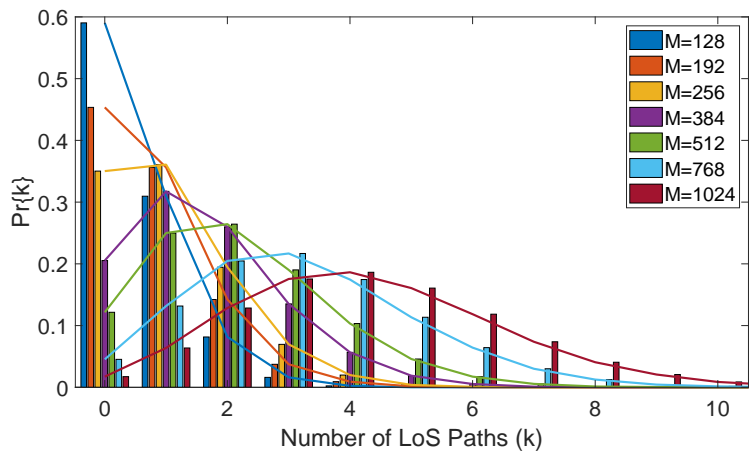}
	\caption{Probability of $k$ LoS paths between the APs and a user with single antenna APs}
	\label{fig:PMF}
\end{figure}

In Fig.~\ref{fig:PMF} we plot the probability mass function~(PMF) of an arbitrarily placed UE having $k$ LoS paths to/ from different APs for different number of single antenna APs~(M). We observe that increasing the number of APs from 128 to 1024 dramatically increases the probability existence of at least one LoS path to the user from approximately 40\% to well above 95\%, thereby resulting in an improved overall LoS coverage. The deterministic nature and high gain of these LoS channels encourage us to deploy more single antenna APs instead of a smaller number of APs with multiple antennas. 
\begin{figure}[t]
	\centering
	\includegraphics[width=0.7\textwidth]{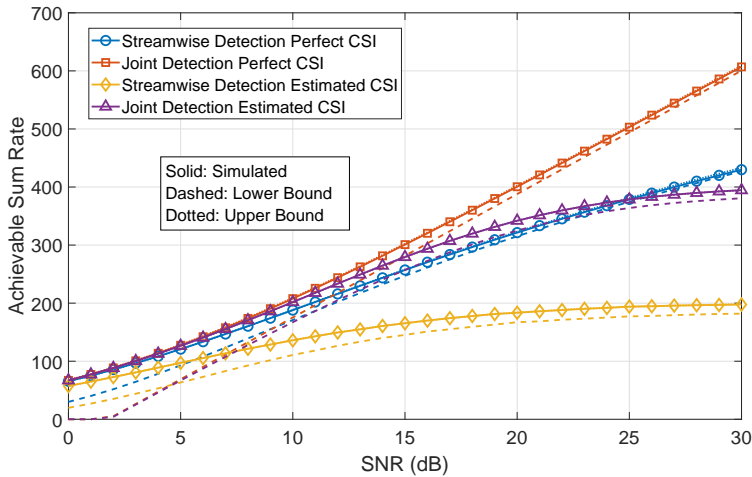}
	\caption{Performance of a 1024 AP CF massive MIMO system under LoS/NLoS channels under different conditions compared with the theoretically evaluated bounds.}
	\label{fig:bounds}
\end{figure}

In Fig.~\ref{fig:bounds} we plot the achievable sum rate of a 64 user CF massive MIMO system with 1024 single antenna APs, under four different settings entailing stream-wise detection and joint detection under both perfect CSI, and estimated CSI. In order to better quantify the effect of channel estimation errors on the system, we assume a received pilot SNR of 20~dB for all the channel coefficients. We observe that in all the cases, the upper bound, obtained using the Jensen's inequality, approximates the achievable rate almost exactly, and the lower bound follows it closely. This fact can be mainly attributed to the dominance of the LoS component(s) in the effective channels to different users, resulting in effectively deterministic channels to/ from all the users. We also note that the sum rates achievable via stream-wise decoding with conjugate beamforming and under perfect CSI are not only significantly lower than those achievable with joint decoding, but also show a saturation effect at higher data SNRs, thus pointing towards an interference limited system. 

\begin{figure}[t]
	\centering
	\includegraphics[width=0.7\textwidth]{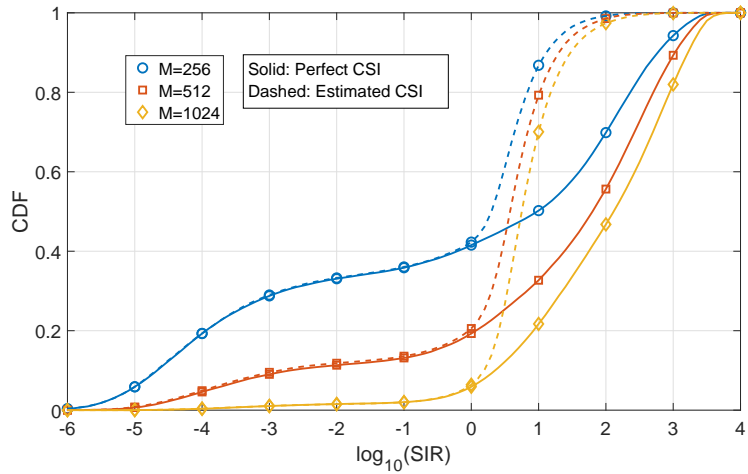}
	\caption{CDF of the achievable SIR for differnet users in a 1024 AP CF massive MIMO system under LoS/NLoS channels with conjugate beamforming based combining under different conditions.}
	\label{fig:CDF_MRC}
\end{figure}

To further investigate the distribution of rates among various users, we plot the cumulative distribution function~(CDF) of the signal to interference ratio~(SIR) achievable by different users in a 64 user system (averaged over 1000 realizations) in Fig.~\ref{fig:CDF_MRC}. We observe that even under perfect/ accurate CSI, a fraction of users approximately equal to the fraction of users without an LoS path, experiences SIRs below $0$ dB indicating a ``capture" of the system by the users with LoS paths. This necessitates the use of interference cancellation at the CPU, motivating the centralized MMSE technique discussed in Section~V.   

\begin{figure}[t]
	\centering
	\includegraphics[width=0.7\textwidth]{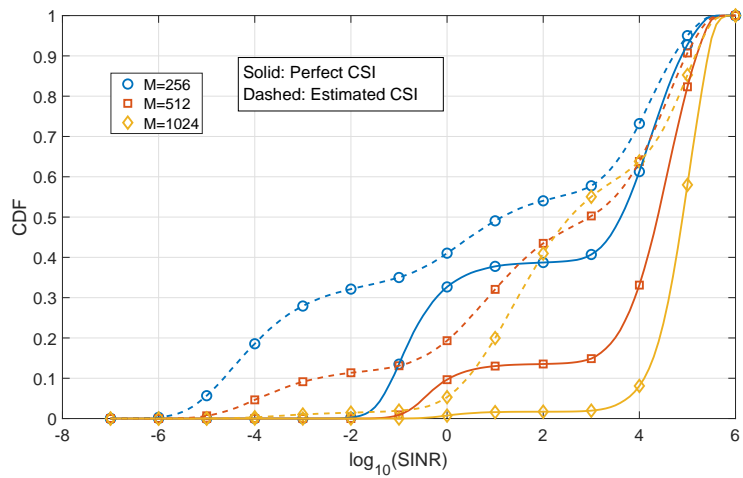}
		\caption{CDF of the achievable SIR for differnet users in a 1024 AP CF massive MIMO system under LoS/NLoS channels with MMSE combining under different conditions.}
	\label{fig:CDF_MMSE}
\end{figure}

In the absence of closed form expressions describing the performance of centralized MMSE combining for the system under observation, we first plot the CDF of achievable signal to interference plus noise ratio~(SINR) for a data SNR of 50~dB in Fig~\ref{fig:CDF_MMSE}. We observe that the CDF shifts significantly towards right for both accurate and estimated CSI at the CPU. Note that, instead of the SIR considered in Fig~\ref{fig:CDF_MRC}, that assumes infinite data SNR, we consider the SINR at a data SNR of 50dB, and hence, under exact CSI, the CDF peaks around an SINR of 50~dB, showing a noise limited channel as opposed to the interference-limited channel, observed in Fig.~\ref{fig:CDF_MRC}.  

\begin{figure}[t]
	\centering
	\includegraphics[width=0.7\textwidth]{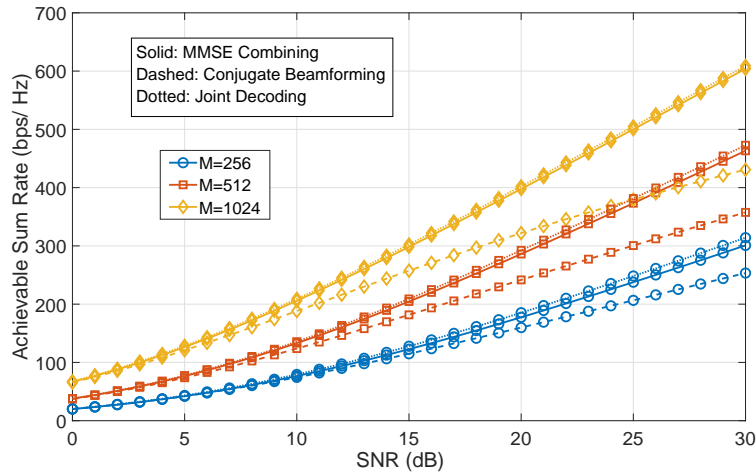}
	\caption{Performance of a 1024 AP CF massive MIMO system under LoS/NLoS channels with accurate CSI for different combining schemes.}
	\label{fig:MMSE_noer}
\end{figure}

\begin{figure}[t]
	\centering
	\includegraphics[width=0.7\textwidth]{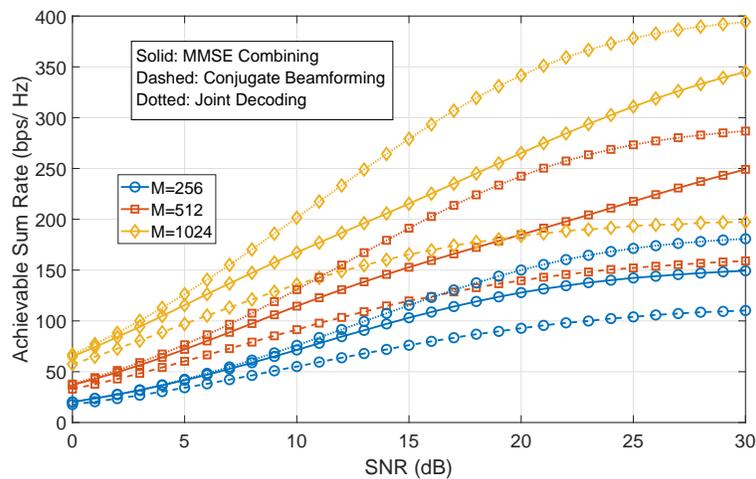}
	\caption{Performance of a 1024 AP CF massive MIMO system under LoS/NLoS channels with estimated CSI for different combining schemes.}
	\label{fig:MMSE_er}
\end{figure}

In Figures~\ref{fig:MMSE_noer} and~\ref{fig:MMSE_er}, we compare the sum rates achievable by a 1024 AP CF massive MIMO system, using the three combining schemes discussed in this paper for both accurate and estimated CSI. In Fig.~\ref{fig:MMSE_noer} we observe that under accurate CSI, the rates achievable with MMSE combining closely approximate the rates achievable using joint detection at the CPU, indicating accurate cancellation of interfering data streams. However as seen in Fig.~\ref{fig:MMSE_er}, the difference between the rates achievable by MMSE combining and joint decoding becomes significant in case of estimated CSI being avalibale at the CPU. This shows that we require better algorithms for CSI acquisition, as well as interference cancellation and data detection, in order to achieve the true potential of CF massive MIMO systems.

\begin{figure}[t]
	\centering
	\includegraphics[width=0.7\textwidth]{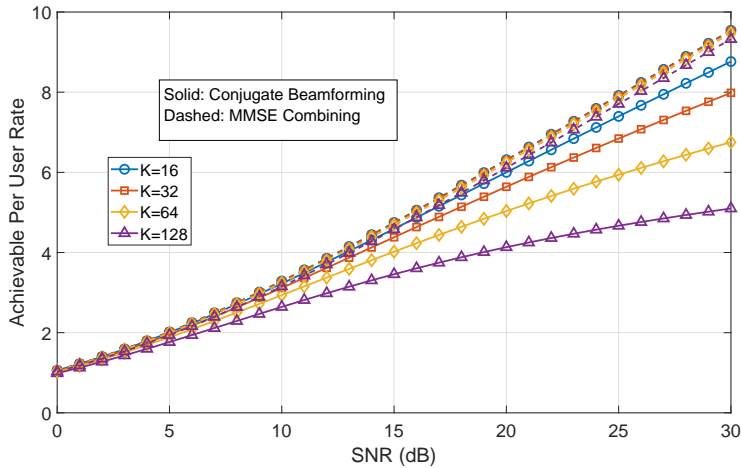}
	\caption{Performance of a 1024 AP CF massive MIMO system under LoS/NLoS channels with accurate CSI for different combining schemes and different numbers of UEs.}
	\label{fig:vsK}
\end{figure}

In Fig.~\ref{fig:vsK} we compare the per user achievable data rates in centralized MMSE and conjugate beamforming for different number of UEs under the availability of accurate CSI at the CPU. Here, we consider only the case with accurate CSI at the CPU/ APs to better isolate the effect of inter-user interference on such systems. It is observed that the per user data rate with MMSE combining remains largely unaffected by the increase in the number of users, whereas an increase in the number of users with conjugate beamforming based combining adversely affects the per user rate. This is due to the increased inter-user interference in the latter case, and is in accordance with intuition.

\begin{figure}[t]
	\centering
	\includegraphics[width=0.7\textwidth]{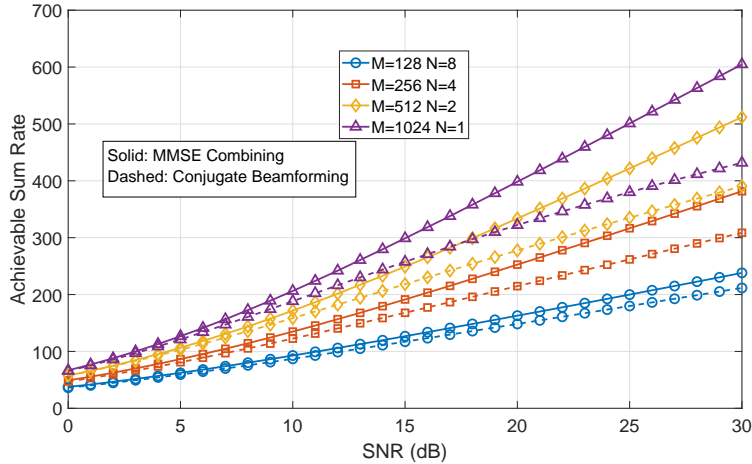}
	\caption{Performance of a CF massive MIMO system under LoS/NLoS channels with accurate CSI for different combining schemes and AP configurations.}
	\label{fig:vsM}
\end{figure}

In Fig.~\ref{fig:vsM} we plot the achievable sum rate for a 64 UE system for different AP configurations, while keeping the antenna density over the area of interest fixed. We observe that an increase in the number of APs with a fixed antenna density results in an almost threefold increase in the achievable sum rates. This justifies our conjecture in Fig.~\ref{fig:PMF} that increasing the number of APs implies an increase in the number of LoS paths which would result in an improvement in the achievable data rates. Here it is pertinent to note that the difference between the performance of centralized MMSE and conjugate beamforming based combining increases dramatically with an increase in the AP density, thereby indicating that the results derived for NLoS channels (cellular massive MIMO) cannot be trivially extended to LoS channels (CF massive MIMO).

\section{Conclusions}
In this paper, we discussed the performance of CF massive MIMO systems under the possibility of probabilistic LoS/ NLoS channels. Using well understood bounds on the achievable rate, we evaluated the performance of these systems under both conjugate beamforming based decoding and joint decoding at the CPU. We observed that in the presence of LoS channels, conjugate beamforming based data detection performs much worse than joint data detection due to the increased inter user interference. We also evaluated the performance of centralized MMSE based data detection at the CPU. We observed that under the availability of accurate CSI at the CPU, MMSE based data detection closely approximates joint detection performance, whereas in the presence of estimated CSI, the performance gap between MMSE detection and joint detection increases. This shows that we require better algorithms for CSI acquisition, as well as interference cancellation and data detection, in order to achieve the true potential of CF massive MIMO systems.

\begin{appendix}

\subsection{The Statistics of the $\mathbf{G}$ matrix}\label{App:A}
In the following, we provide the statistics of diagonal and off-diagonal elements of $\mathbf{G}$. Here, we first expand the diagonal element $g_{kk}$, and then derive its mean and variance as given below.
\begin{multline}
	g_{kk}=\mathbf{h}_k^H\mathbf{h}_k=((\bm{\partial}_k\otimes\mathbf{1}_N)\odot\mathbf{\bar{h}}_k+(\bm{\beta}_k\otimes\mathbf{1}_N)\odot\mathbf{\dot{h}}_k)^H((\bm{\partial}_k\otimes\mathbf{1}_N)\odot\mathbf{\bar{h}}_k+(\bm{\beta}_k\otimes\mathbf{1}_N)\odot\mathbf{\dot{h}}_k)\\
	=\sum_{m=1}^{M}\delta_{mk}\sum_{i=1}^N\lvert\bar{h}_{mi,k}\rvert^2+\sum_{m=1}^M\beta_{mk}\sum_{i=1}^N|\dot{h}_{mi,k}|^2+2\sum_{m=1}^M\delta_{mk}\sqrt{\beta_{mk}}\sum_{i=1}^N\Re\{\bar{h}_{mi,k}\dot{h}_{mi,k}^*\}\\=\underbrace{N\sum_{m=1}^M\delta_{mk}G_kG_m\left(\frac{\ell_{k}^{'}\ell_m}{4\pi x_{mk}}\right)^2}_{\triangleq \, g_{kk}^{(1)}}+\underbrace{\sum_{m=1}^M\beta_{mk}\sum_{i=1}^N|\dot{h}_{mi,k}|^2}_{\triangleq \, g_{kk}^{(2)}}+\underbrace{2\sum_{m=1}^M\delta_{mk}\sqrt{\beta_{mk}}\sum_{i=1}^N\Re\{\bar{h}_{mi,k}\dot{h}_{mi,k}^*\} }_{\triangleq \, g_{kk}^{(3)}} \, ,
	\label{eq:gkkexpand}
\end{multline}
We can now write $|g_{kk}|^2=|g^{(1)}_{kk}|^2+|g^{(2)}_{kk}|^2+|g^{(3)}_{kk}|^2+2\Re\{g^{(1)}_{kk}g^{(2)*}_{kk}\} +2\Re\{g^{(3)}_{kk}g^{(2)*}_{kk}\}+2\Re\{g^{(1)}_{kk}g^{(3)*}_{kk}\}.$ From \eqref{eq:gkkexpand}, it is easy to show that $E[g_{kk}^{(1)}g_{kk}^{(3)*}]=E[g_{kk}^{(2)}g_{kk}^{(3)*}]=0$, and 

\begin{equation}
\label{eq:mean_gkk}
	\Eb[g_{kk}]=N\sum_{m=1}^M\left(P_{mk}G_kG_m\left(\frac{\ell_{k}^{'}\ell_m}{4\pi x_{mk}}\right)^2+\beta_{mk}\right), 
\end{equation} 

%%\begin{equation}
%%E[g_{kk}^{(1)}]=NG_k\left(\frac{\ell_{k}^{'}}{4\pi}\right)^2\sum_{m=1}^MP_{mk}G_m\left(\frac{\ell_m}{x_{mk} }\right)^2, 
%%\end{equation}
%%and
\begin{multline}
\hspace{-1 cm} E\left[\left|g_{kk}^{(1)}\right|^2\right]=N^2G_k^2\left(\frac{\ell_{k}^{'}}{4\pi}\right)^4\left(\sum_{m=1}^{M}P_{mk}G_m^2\left(\frac{\ell_m}{x_{mk}}\right)^4+\sum_{m=1}^{M}\sum_{\substack{l=1\\l\neq m}}^{M}P_{mk}P_{lk}G_mG_l\left(\frac{\ell_m\ell_l}{x_{mk}x_{lk}}\right)^2\right),
\end{multline}
%%\begin{equation}
%%E[g_{kk}^{(2)}]=N\sum_{m=1}^M\beta_{mk}, 
%%\end{equation}
\begin{equation}
E\left[\left|g_{kk}^{(2)}\right|^2\right]=N\sum_{m=1}^M\beta_{mk}^2+N^2\sum_{m=1}^M\sum_{l=1}^M\beta_{mk}\beta_{lk},
\end{equation}
%\begin{equation}
%E[g_{kk}^{(3)}]=0, 
%\end{equation}
\begin{equation}
E\left[\left|g_{kk}^{(3)}\right|^2\right]=2NG_k\left(\frac{\ell_{k}^{'}}{4\pi}\right)^2\sum_{m=1}^MP_{mk}\beta_{mk}G_m\left(\frac{\ell_m}{x_{mk} }\right)^2,
\end{equation}
%\begin{equation}
%E[g_{kk}^{(1)}g_{kk}^{(3)*}]=E[g_{kk}^{(2)}g_{kk}^{(3)*}]=0,
%\end{equation}
and finally,
\begin{equation}
E[g_{kk}^{(1)}g_{kk}^{(2)*}]=E[g_{kk}^{(1)}]E^*[g_{kk}^{(2)}]=N^2G_k\left(\frac{\ell_{k}^{'}}{4\pi}\right)^2\sum_{l=1}^M\beta_{lk}\sum_{m=1}^MP_{mk}G_m\left(\frac{\ell_m}{x_{mk} }\right)^2.
\label{eq:gkk2mpart}
\end{equation}
Therefore, from \eqref{eq:mean_gkk}-\eqref{eq:gkk2mpart}, we have
\begin{multline}
	\text{var}(g_{kk})= \E{|g_{kk} - \E{g_{kk}}|^2} = \E{|g_{kk}|^2} - \left(\E{g_{kk}}\right)^2\\
	= N^2G_k^2\left(\frac{\ell_{k}^{'}}{4\pi}\right)^4\sum_{m=1}^MP_{mk}(1-P_{mk})\left(\frac{\ell_m}{x_{mk}}\right)^4+N\sum_{m=1}^M\beta_{mk}^2+2N\sum_{m=1}^MP_{mk}\beta_{mk}G_kG_m\left(\frac{\ell_{k}^{'}\ell_m}{4\pi x_{mk}}\right)^2.
	\label{eq:gkkvar}
\end{multline}
%%Based on this, we can write,
%%\begin{equation}
%%E[g_{kk}]=N\sum_{m=1}^M\left(P_{mk}G_kG_m\left(\frac{\ell_{k}^{'}\ell_m}{4\pi x_{mk}}\right)^2+\beta_{mk}\right), 
%%\end{equation}
%%and
%%\begin{multline}
%%\text{var}(g_{kk})=N^2G_k^2\left(\frac{\ell_{k}^{'}}{4\pi}\right)^4\sum_{m=1}^MP_{mk}(1-P_{mk})\left(\frac{\ell_m}{x_{mk}}\right)^4\\+N\sum_{m=1}^M\beta_{mk}^2+4N\sum_{m=1}^MP_{mk}\beta_{mk}G_kG_m\left(\frac{\ell_{k}^{'}\ell_m}{4\pi x_{mk}}\right)^2.
%%\end{multline}

%From~\eqref{eq:gkkvar} it is clear that the upper and the lower bound on the achievable rate converge as $N\to\infty$, and $P_{mk}$ moves closer to either 0 or 1. However, due to the distributed nature of the system, no such easy conclusions can be derived as $M\to\infty$. 

We next evaluate the fourth moment of $g_{kk}$ as follows.
\begin{align}
\label{eq:gkk4moment}
\nonumber \E{\left| g_{kk} \right|^4}  & = \sum\limits_{i=1}^{3} \E{\left| g_{kk}^{(i)} \right|^4} + 6\sum\limits_{i=1}^{3}\sum\limits_{j=1, j \neq i}^{3} \E{ \left| g_{kk}^{(i)} \right|^2 \left| g_{kk}^{(j)} \right|^2 } + 4 \E{ \left| g_{kk}^{(1)} \right|^3    }\E{  g_{kk}^{(2)}  }\\ 
& + 4 \E{ \left| g_{kk}^{(2)} \right|^3    } \E{   g_{kk}^{(1)} } + 12 \E{ \left| g_{kk}^{(3)} \right|^2  g_{kk}^{(2)}  } \E{  g_{kk}^{(1)} } \, ,
\end{align} 
\noindent where
\begin{align}
\label{eq:compgkk4m}
\nonumber \E{\left| g_{kk}^{(1)} \right|^4} & = N^4 \left[\substack{\sum\limits_{ \substack{ m_i \in \{1,2,...,M\}; m_i \neq m_j;\\ i = \{1,2,3,4\} ;  j = \{1,2,3,4\} } } \left( \prod_{l=1}^{4}a_{m_l k}P_{m_l k} \right)\,+ 6\sum\limits_{m=1}^{M} a_{mk}^4 P_{mk} \\+ \sum\limits_{ \substack{ m_i \in \{1,2,...,M\};\\ i = \{1,3,4\};\\m_1 \neq m_3 \neq m_4  } } P_{m_1k}P_{m_3k}P_{m_4k}a_{m_1k}^2a_{m_3k}a_{m_4k} \\+ \sum\limits_{ \substack{ m_i \in \{1,2,...,M\};\\ i = \{1,2,4\};\\m_1 \neq m_2 \neq m_4  } }P_{m_1k}P_{m_2k}P_{m_4k}(a_{m_1k}^2a_{m_2k}a_{m_4k}+a_{m_1k}a_{m_2k}^2a_{m_4k})\\+ \sum\limits_{ \substack{ m_i \in \{1,2,...,M\};\\ i = \{1,2,3\};\\m_1 \neq m_2 \neq m_3  } }P_{m_1k}P_{m_2k}P_{m_3k}\left( a_{m_1k}^2a_{m_2k}a_{m_3k} + a_{m_1k}a_{m_2k}^2a_{m_3k}+a_{m_1k}a_{m_2k}a_{m_3k}^2\right)}  \right. \\
&\left. + \left.\substack{ 2\sum\limits_{ \substack{ m_i \in \{1,2,...,M\};\\ i = \{1,2\};\\ m_1 \neq m_2  } } P_{m_1k}P_{m_2k}\left( a_{m_1k}^2 a_{m_2k}^2 + a_{m_1k}a_{m_2k}^3+a_{m_1k}^3 a_{m_2k}\right) \, + \, \sum\limits_{ \substack{ m_i \in \{1,2,...,M\};\\ i = \{1,4\};\\ m_1 \neq m_4  } } 2a_{m_1k}^3 a_{m_4k} \\ + \sum\limits_{ \substack{ m_i \in \{1,2,...,M\};\\ i = \{1,3\};\\ m_1 \neq m_3  } } P_{m_1k}P_{m_3k}\left( a_{m_1k}^2 a_{m_3k}^2 + 2a_{m_1k}^3 a_{m_3k}\right) }\right\} \right] ,\\
\nonumber \E{\left| g_{kk}^{(2)} \right|^4} & = \sum\limits_{\substack{ m_i \in \{1,2,...,M\};\\ i = \{1,2,3,4\} }} \left(\prod\limits_{i=1}^{4}\beta_{m_ik} \right) N \left[N^3 + 12N^2 + 104N +513 \right]\\
\nonumber \E{\left| g_{kk}^{(3)} \right|^4} & = N^2\left[\left\{ \substack{18\sum\limits_{m=1}^{M} a_{mk}^2\beta_{mk}^2 P_{mk} + 2\sum\limits_{ \substack{ m_i \in \{1,2,...,M\};\\ i = \{1,2\};\\ m_1 \neq m_2  } } P_{m_1k}P_{m_2k} a_{m_1k} a_{m_2k} \beta_{m_1k} \beta_{m_2k}\\ + \sum\limits_{ \substack{ m_i \in \{1,2,...,M\};\\ i = \{1,3\};\\ m_1 \neq m_3  } } P_{m_1k}P_{m_3k} a_{m_1k} a_{m_3k} \beta_{m_1k} \beta_{m_3k} + }\right\} \right], \\
\E{\left| g_{kk}^{(1)} \right|^3} & = N^3 \left[\substack{ \sum\limits_{ \substack{ m_i \in \{1,2,...,M\}; m_i \neq m_j;\\ i = \{1,2,3\} ;  j = \{1,2,3\} } } \left( \prod_{l=1}^{3}a_{m_l k}P_{m_l k} \right)\,+ 2\sum\limits_{m=1}^{M} a_{mk}^3 P_{mk} \\+ \sum\limits_{ \substack{ m_i \in \{1,2,...,M\};\\ i = \{1,2\};\\ m_1 \neq m_2  } } P_{m_1k}P_{m_2k}\left( a_{m_1k}a_{m_2k}^2+a_{m_1k}^2 a_{m_2k}\right)       }\right] \\
\E{\left| g_{kk}^{(2)} \right|^3} & = \sum\limits_{\substack{ m_i \in \{1,2,...,M\};\\ i = \{1,2,3\} }} \left(\prod\limits_{i=1}^{3}\beta_{m_ik} \right) N\left[N^2+3N+26\right] \\
\E{\left| g_{kk}^{(2)} \right|^2\left| g_{kk}^{(3)} \right|^2} & = 2N\sum\limits_{m=1}^{M}\beta_{mk}a_{mk} P_{mk} \left[\substack{ \sum\limits_{\substack{ m_i \in \{1,2,...,M\};\\ i = \{1,2\};\\ m_1 \neq m_2}}\beta_{m_1k}\beta_{m_2k}N(N-1) \\+\sum\limits_{\substack{ m_i \in \{1,2,...,M\};\\ i = \{1,2\}}} N \beta_{m_1k}\beta_{m_2k}+ \sum\limits_{m_1=1}^{M}\beta_{m_1k}^2N(N-1)+ 3N\sum\limits_{m_1=1}^{M}\beta_{m_1k}^2      }\right] \\
\E{\left| g_{kk}^{(3)} \right|^2 g_{kk}^{(2)} } & = \frac{1}{4}\left[\substack{\sum\limits_{\substack{ m_i \in \{1,2,...,M\};\\ i = \{1,2\}}}N^2\beta_{m_1k}\beta_{m_2k}a_{m_2k}P_{m_2k} \\ + \sum\limits_{\substack{ m_i \in \{1,2,...,M\};\\ i = \{1,2\};\\ m_1 \neq m_2}}N \beta_{m_1k}\beta_{m_2k}a_{m_2k}P_{m_2k} + \sum\limits_{m=1}^{M}3N\beta_{m_1k}\beta_{mk}^2a_{mk}P_{mk}   }\right] \, .
\label{eq:compgkk4me}
\end{align}
\noindent Using \eqref{eq:mean_gkk}-\eqref{eq:gkk4moment}, we can evaluate the variance of $|g_{kk}|^2$ as $\var{|g_{kk}|^2} = \E{|g_{kk}|^4} - \left(\E{|g_{kk}|^2}\right)^2$.
Similarly, the off-diagonal element $g_{kl}$ is given by
\begin{multline}
	g_{kl}=\mathbf{h}_k^H\mathbf{h}_l=((\bm{\partial}_k\otimes\mathbf{1}_N)\odot\mathbf{\bar{h}}_k+(\bm{\beta}_k\otimes\mathbf{1}_N)\odot\mathbf{\dot{h}}_k)^H((\bm{\partial}_l\otimes\mathbf{1}_N)\odot\mathbf{\bar{h}}_l+(\bm{\beta}_l\otimes\mathbf{1}_N)\odot\mathbf{\dot{h}}_l)\\
	=\sum_{m=1}^{M}\delta_{mk}\delta_{ml}\sum_{i=1}^N\bar{h}_{mi,k}\bar{h}_{mi,l}+\sum_{m=1}^M\sqrt{\beta_{mk}\beta_{ml}}\sum_{i=1}^N|\dot{h}_{mi,k}|^2+2\sum_{m=1}^M\delta_{mk}\sqrt{\beta_{ml}}\sum_{i=1}^N\Re\{\bar{h}_{mi,k}\dot{h}_{mi,l}^*\}\\
	=\sum_{m=1}^M\delta_{ml}\delta_{mk}G_m\sqrt{G_kG_l}\left(\frac{\ell_{k}^{'}\ell_l\ell_m^2}{16\pi^2 x_{mk}x_{ml}}\right)e^{\iota\frac{2\pi}{\lambda_c}(x_{mk}-x_{ml})}\sum_{i=1}^N e^{\iota \frac{2\pi d}{\lambda_c}i(\sin(\theta_{mk})-sin(\theta_{ml}))} \\
	+\sum_{m=1}^M\sqrt{\beta_{mk}\beta_{ml}}\sum_{i=1}^N\dot{h}^*_{mi,l}\dot{h}_{mi,k}+2\sum_{m=1}^M\delta_{mk}\sqrt{\beta_{ml}}\sum_{i=1}^N\Re\{\bar{h}_{mi,k}\dot{h}_{mi,l}^*\}=g_{kl}^{(1)}+g_{kl}^{(2)}+g_{kl}^{(3)},
	\label{eq:gkldef}
\end{multline}
where $g_{kl}^{(i)}$ ($i = 1,2,3$) are defined implicitly. From \eqref{eq:gkldef}, we have
%%\begin{equation}
%%E[g_{kl}^{(1)}]=\sum_{m=1}^MP_{ml}P_{mk}G_m\sqrt{G_kG_l}\left(\frac{\ell_{k}^{'}\ell_l\ell_m^2}{16\pi^2 x_{mk}x_{ml}}\right)e^{\iota\frac{2\pi}{\lambda_c}(x_{mk}-x_{ml})}\sum_{i=1}^N e^{\iota \frac{2\pi d}{\lambda_c}i(\sin(\theta_{mk})-sin(\theta_{ml}))}
%%\end{equation}
%%\begin{equation}
%%E\left[\left|g_{kl}^{(3)}\right|^2\right]=4NG_k\left(\frac{\ell_{k}^{'}}{4\pi}\right)^2\sum_{m=1}^MP_{mk}\beta_{ml}G_m\left(\frac{\ell_m}{x_{mk} }\right)^2.
%%\end{equation}
%%Consequently, we can write
\begin{equation}
	\Eb[g_{kl}]=\sum_{m=1}^MP_{ml}P_{mk}G_m\sqrt{G_kG_l}\left(\frac{\ell_{k}^{'}\ell_l\ell_m^2}{16\pi^2 x_{mk}x_{ml}}\right)e^{\iota\frac{2\pi}{\lambda_c}(x_{mk}-x_{ml})}\sum_{i=1}^N e^{\iota \frac{2\pi d}{\lambda_c}i(\sin(\theta_{mk})-sin(\theta_{ml}))},
\end{equation}
and 
\begin{multline}
	\text{var}(g_{kl})=\sum_{m=1}^MP_{ml} (1-P_{ml}) P_{mk} (1-P_{mk}) G_m\sqrt{G_kG_l}\left(\frac{\ell_{k}^{'}\ell_l\ell_m^2}{16\pi^2 x_{mk}x_{ml}}\right)e^{\iota\frac{2\pi}{\lambda_c}(x_{mk}-x_{ml})}\\ \times \sum_{i=1}^N e^{\iota \frac{2\pi d}{\lambda_c}i(\sin(\theta_{mk})-sin(\theta_{ml}))}+N\sum_{m=1}^M\beta_{mk}\beta_{ml}+4NG_k\left(\frac{\ell_{k}^{'}}{4\pi}\right)^2\sum_{m=1}^MP_{mk}\beta_{ml}G_m\left(\frac{\ell_m}{x_{mk} }\right)^2.
\end{multline}

%%%%%%%%%%%%%%%%%%%%%%%%%%%%%%%%%%%%%%%%%%%%%%%%%%%%%%%%%%%%%%

\subsection{Variance of $z_k$}\label{App:B}
We know that $z_k=\sqrt{N}_0\mathbf{h}_k^H\mathbf{w}$. Consequently, we have $\text{var}({z_k})=N_0\mathbf{h}_k^H\E{\mathbf{w}\mathbf{w}^H}\mathbf{h}_k$, i.e.,
\begin{equation}
	\text{var}({z_k})=N_0\mathbf{h}_k^H\mathbf{I}_{MN}\mathbf{h}_k\\=N_0N\sum_{m=1}^M\left(P_{mk}G_kG_m\left(\frac{\ell_{k}^{'}\ell_m}{4\pi x_{mk}}\right)^2+\beta_{mk}\right).
\end{equation}

%%%%%%%%%%%%%%%%%%%%%%%%%%%%%%%%%%%%%%%%%%%%%%%%%%%%%%%%%%%%%%

\subsection{Statistics of the Effective Channel under Estimated CSI}\label{App:C}
We know that the true channel between the $k$th UE and the $m$th AP, $\mathbf{h}_{mk}$ can be written in terms of its estimate as in~\eqref{eq:ch_est}. Therefore, 
\begin{equation}
	\mathbf{\hat{h}}^H_{mk}\mathbf{h}_{mk}=\mathbf{\hat{h}}^H_{mk}\mathbf{C}^{\frac{1}{2}}_{mk}\mathbf{\hat{h}}_{mk}+\mathbf{\hat{h}}^H_{mk}\mathbf{C}_{mk}\mathbf{\tilde{h}}_{mk}=\hat{g}_{m,kl}+\tilde{g}_{m,kl},
\end{equation}
Consequently, we have $\hat{g}_{kk}=\sum\limits_{m=1}^{M} \hat{g}_{m,kl} = \sum\limits_{m=1}^M\mathbf{\hat{h}}^H_{mk}\mathbf{C}^{\frac{1}{2}}_{mk}\mathbf{\hat{h}}_{ml}$. Clearly,
\begin{equation}
	\E{\hat{g}_{kk}}=G_k\left(\frac{\ell'_k}{4\pi} \right)^2\sum_{m=1}^M P_{mk}G_m\left(\frac{\ell_m}{x_{mk}} \right)^2 \mathbf{a}^H(\theta_{mk})\mathbf{\check{C}}_{mk}^{\frac{1}{2}}\mathbf{a}(\theta_{mk})+\sum_{m=1}^M\beta_{mk}\text{Tr}\{\E{\mathbf{C}_{mk}}\},
\end{equation}
where
\begin{multline}
	\mathbf{\check{C}}_{mk}=\mathcal{E}_p \left( G_mG_k\left( \frac{\ell'_k\ell_m}{x_{mk}} \right)^2\mathbf{a}(\theta_{mk})\mathbf{a}^{H}(\theta_{mk})+\beta_{mk}\mathbf{I}_N \right)\\\times\left(\mathcal{E}_p G_mG_k\left( \frac{\ell'_k\ell_m}{x_{mk}} \right)^2\mathbf{a}(\theta_{mk})\mathbf{a}^{H}(\theta_{mk})+(\beta_{mk}\mathcal{E}_p+N_0)\mathbf{I}_N \right)^{-1}\\ \times \left( G_mG_k\left( \frac{\ell'_k\ell_m}{x_{mk}} \right)^2\mathbf{a}(\theta_{mk})\mathbf{a}^{H}(\theta_{mk})+\beta_{mk}\mathbf{I}_N \right),
\end{multline}
and 
\begin{equation}
	\E{\mathbf{C}_{mk}^{\frac{1}{2}}}=\mathbf{\check{C}}^{\frac{1}{2}}_{mk}P_{mk}+(1-P_{mk})\frac{\beta_{mk}\mathcal{E}_p}{\beta_{mk}\mathcal{E}_p+N_0}\mathbf{I}_N.
\end{equation}
Here, $\mathbf{\check{C}}_{mk}$ represents the matrix $\mathbf{C}_{mk}$ in case an LoS link exists between the $m$th AP and the $k$th UE. We can now write $	\lvert \hat{g}_{kk} \rvert^2 =\sum\limits_{m=1}^M\sum\limits_{l=1}^M\mathbf{\hat{h}}^H_{mk}\mathbf{C}^{\frac{1}{2}}_{mk}\mathbf{\hat{h}}_{mk}\mathbf{\hat{h}}^T_{ml}\mathbf{C}^{\frac{1}{2}}_{mk}\mathbf{\hat{h}}^{*}_{ml}$, and therefore, we have

%\begin{equation}
%	\lvert \hat{g}_{kk} \rvert^2 =\sum_{m=1}^M\sum_{l=1}^M\mathbf{\hat{h}}^H_{mk}\mathbf{C}^{\frac{1}{2}}_{mk}\mathbf{\hat{h}}_{mk}\mathbf{\hat{h}}^T_{ml}\mathbf{C}^{\frac{1}{2}}_{mk}\mathbf{\hat{h}}^{*}_{ml}, 
%	\label{eq:g_kkhat2}
%\end{equation}
%and consequently,
\begin{multline}
	\E{\lvert \hat{g}_{kk} \rvert^2}=G^2_k\left(\frac{\ell'_k}{4\pi} \right)^4\sum_{m=1}^M P_{mk}G_m^2\left(\frac{\ell_m}{x_{mk}} \right)^4 (\mathbf{a}^H(\theta_{mk})\mathbf{\check{C}}_{mk}\mathbf{a}(\theta_{mk}))^2+\sum_{m=1}^M\beta^2_{mk}\text{Tr}^2\{\E{\mathbf{C}_{mk}}\}\\+G^2_k\left(\frac{\ell'_k}{4\pi} \right)^4\sum_{m=1}^M\sum_{\substack{l=1\\l\neq m}}^M P_{ml}P_{mk}G_mG_l\left(\frac{\ell_m}{x_{mk}} \right)^2\left(\frac{\ell_l}{x_{lk}} \right)^2 (\mathbf{a}^H(\theta_{mk})\mathbf{\check{C}}_{mk}\mathbf{a}(\theta_{mk}))(\mathbf{a}^H(\theta_{lk})\mathbf{\check{C}}_{lk}\mathbf{a}(\theta_{lk}))\\+\sum_{m=1}^M\sum_{l=1}^M\beta_{mk}\beta_{lk}\text{Tr}\{\E{\mathbf{C}_{mk}}\}\text{Tr}\{\E{\mathbf{C}_{lk}}\}\\+4G_k\left(\frac{\ell'_k}{4\pi} \right)^2\sum_{m=1}^M P_{mk}G_m\beta_{mk}\left(\frac{\ell_m}{x_{mk}} \right)^2 \mathbf{a}^H(\theta_{mk})\mathbf{\check{C}}_{mk}\mathbf{a}(\theta_{mk})\text{Tr}\{\mathbf{\check{C}}_{mk}\}.
\end{multline}

Similarly, we have $\tilde{g}_{kk}=\sum\limits_{m=1}^M \tilde{g}_{m,kl}= \sum\limits_{m=1}^M\mathbf{\hat{h}}^H_{mk}\mathbf{\bar{C}}^{\frac{1}{2}}_{mk}\mathbf{\tilde{h}}_{ml}$, and $\lvert \tilde{g}_{kk}\rvert^2 =\sum\limits_{m=1}^M\sum\limits_{l=1}^{M}\mathbf{\hat{h}}^H_{mk}\mathbf{\bar{C}}_{mk}\mathbf{\tilde{h}}_{ml}\mathbf{\tilde{h}}^H_{mk}\mathbf{\bar{C}}^{\frac{1}{2}}_{lk}\mathbf{\hat{h}}_{ml}$. Therefore, we can write

\begin{equation}
	\E{\lvert \tilde{g}_{kk}\rvert^2} =\sum_{m=1}^M\sum_{l=1}^{M}\E{\mathbf{\hat{h}}^H_{mk}\mathbf{\bar{C}}^{\frac{1}{2}}_{mk}\E{\mathbf{\tilde{h}}_{ml}\mathbf{\tilde{h}}^H_{mk}}\mathbf{\bar{C}}^{\frac{1}{2}}_{lk}\mathbf{\hat{h}}_{lk} }  =\sum_{m=1}^M\E{\mathbf{\hat{h}}^H_{mk}\mathbf{\bar{C}}_{mk}\mathbf{\hat{h}}_{mk} } \, .
\end{equation}
%%that can be reduced to
%%\begin{equation}
%%\E{\lvert \tilde{g}_{kk}\rvert^2} =\sum_{m=1}^M\E{\mathbf{\hat{h}}^H_{mk}\mathbf{\bar{C}}_{mk}\mathbf{\hat{h}}_{mk} }.
%%\end{equation}
\end{appendix}

\bibliographystyle{IEEEtran}
\bibliography{bibJournalList,fading}
\end{document}